\documentclass[12pt]{article}
\makeatletter\renewcommand{\paragraph}{\@startsection{paragraph}{4}%
  {\z@}{1ex \@plus 0.6ex \@minus .4ex}{-1em}{\normalfont\normalsize\bfseries\boldmath}}
\makeatother

\usepackage[utf8]{inputenc}
\usepackage[margin=1in]{geometry}
\usepackage{amsmath}
\usepackage{amsthm}
\usepackage{booktabs}
\usepackage{xcolor}
\usepackage{hyperref}
\usepackage{graphicx}
\usepackage{algorithm} 
\usepackage[noend]{algpseudocode} 
\usepackage{thmtools} 
\usepackage[capitalize]{cleveref}
\usepackage{authblk}
\crefname{equation}{}{} 
\Crefname{equation}{}{} %
\usepackage{tikz}

\usepackage{amssymb,amsfonts}

\def\eps{\varepsilon}           
\def\epstr{\epsilon}            
\def\nq{\hspace{-1em}}          
\def\fr#1#2{{\textstyle{\frac{#1}{#2}}}} 

\DeclareMathOperator*{\argmax}{argmax}

\def\cA{\mathcal{A}}
\def\cE{\mathcal{E}}
\def\cG{\mathcal{G}}
\def\cM{\mathcal{M}}
\def\cP{\mathcal{P}}
\def\cT{\mathcal{T}}
\def\cQ{\mathcal{Q}}

\algnewcommand\algorithmicinput{\textbf{Input:}}
\algnewcommand\Input{\item[\algorithmicinput]}
\algnewcommand\algorithmicoutput{\textbf{Output:}}
\algnewcommand\Output{\item[\algorithmicoutput]}
\algnewcommand\algorithmiceffect{\textbf{Effect:}}
\algnewcommand\Effect{\item[\algorithmiceffect]}

\def\h{\text\ae}
\def\rO{r$O$} 
\def\rOc{r$O$-computable} 
\def\rOi{_\text{refl}^O}
\def\TS{{T\!S}} 

\DeclareMathOperator{\flip}{flip}
\DeclareMathOperator{\eval}{eval}
\DeclareMathOperator{\query}{query}

\newtheoremstyle{boldtitle}{\topsep}{\topsep}{}{}{\bfseries\boldmath}{}{5pt}{\thmname{#1}\thmnumber{ #2}\thmnote{ (#3)}}
\theoremstyle{boldtitle}
\newtheorem{theorem}{Theorem}
\newtheorem{definition}[theorem]{Definition}
\newtheorem{lemma}[theorem]{Lemma}

\title{\bf\Large\hrule height5pt \vskip 4mm
Limit-Computable Grains of Truth for \\ Arbitrary Computable Extensive-Form (Un)Known Games
\vskip 4mm \hrule height2pt}

\begin{document}

\author[1]{\bf Cole Wyeth}
\affil[1]{David R. Cheriton School of Computer Science, University of Waterloo}
\author[2]{\bf Marcus Hutter}
\affil[2]{Google DeepMind and Australian National University}
\author[3]{\bf Jan Leike}
\affil[3]{Anthropic}
\author[4]{\bf Jessica Taylor}
\affil[4]{Median Group}
\date{September 2024} 
\maketitle

\begin{abstract}
    A Bayesian player acting in an infinite multi-player game
    learns to predict the other players' strategies
    if his prior assigns positive probability to their play (or contains a \emph{grain of truth}).
    Kalai and Lehrer's classic \emph{grain of truth problem} is to find a reasonably large class of strategies
    that contains the Bayes-optimal policies with respect to this class,
    allowing mutually-consistent beliefs about strategy choice that obey the rules of Bayesian inference.
    Only small classes are known to have a grain of truth
    and the literature contains several related impossibility results.
    In this paper we present a formal and general solution to
    the full grain of truth problem:
    we construct a class of strategies wide enough to contain all computable strategies as well as Bayes-optimal strategies for every reasonable  prior over the class.
    When the ``environment'' is a known repeated stage game, we show convergence in the sense of \cite{kalai_rational_1993} and \cite{kalai_subjective_1993}.
    When the environment is unknown, agents using Thompson sampling
    converge to play $\varepsilon$-Nash equilibria
    in arbitrary unknown computable multi-agent environments.
    Finally, we include an application to self-predictive policies that avoid planning.
    While these results use computability theory only as a conceptual tool to solve a classic game theory problem,
    we show that our solution can naturally be computationally approximated
    arbitrarily closely. 

    \newpage\def\contentsname{\centering\normalsize Contents}\setcounter{tocdepth}{1}
    {\parskip=-2.7ex\tableofcontents}
\end{abstract}

\section{Introduction}\label{sec:intro}

We will consider the behavior of Bayesian players engaged in infinite games. This is a core problem of game theory, but because of reflective difficulties the first rich class of solutions grounded in rigorous decision theory has only recently been proposed \cite{leike_formal_2016}. According to the standards of rational behavior derived by von Neumann and Morgenstern \cite{von_neumann_theory_1944}, players should act to maximize their expected utility. This requires each player to maintain and update beliefs about all opponents' strategies, which presumably depend on their reasoning about the player's strategy. This well studied infinite recursion gives rise to the grain of truth problem \cite{kalai_rational_1993}: how can one construct a consistent set of beliefs for each player such that they all assign nonzero probability to each others' Bayes-optimal strategies? This problem is known to be difficult, with many impossibility results (e.g.\ \cite{nachbar_prediction_1997,nachbar_beliefs_2005,foster_impossibility_2001}).

\paragraph{Solution overview.}
Using the recently invented concept of a ``reflective oracle'' \cite{fallenstein_reflective_2015} which selects a commonly-known fixed point for players' mutually recursive belief distributions, one can construct priors that overcome recursive difficulties and are even limit-computable. Intuitively, reflective oracles allow algorithms to make predictions about their own behavior (by asking the oracle), which would normally be impossible by diagonal arguments. In Section~\ref{sec:refl_oracles}, we will introduce reflective oracles and show how they can be used to construct a class of strategies $\cP\rOi$. In Section~\ref{sec:multiagents} we show that when strategies are chosen in $\cP\rOi$, each player's ``subjective environment'' is reflective-oracle computable and has a Bayes-optimal strategy in $\cP\rOi$ (Theorem~\ref{sigma_i_refl} and \cref{thm:O_gen_optimal_strategy}). This allows us to construct an interesting reflective-oracle computable Nash equilibrium (\cref{thm:refl_oracle_comp_nash}) which has the informal (self-referential) interpretation that ``everyone assumes everyone else is playing a best response to this strategy profile.'' To model players' uncertainty about the strategies they face we construct a dominant ``mixture'' strategy $\zeta \in \cP\rOi$ in \cref{sec:convergence_for_Bayesian}\  \cref{thm:dominant_zeta}, which is the final step to establish the grain of truth property. This allows us to satisfy the conditions of \cite{kalai_rational_1993}, proving that Bayesian players with beliefs supported on $\cP\rOi$ will converge to an $\eps$-Nash equilibrum in an infinitely repeated stage game; this is a particularly surprising result because mixed strategy Nash equilbria arise naturally despite the fact that Bayesians are not a priori required to randomize. Finally, we address the general case that players do not know either the game or their opponents' strategies, but only know the classes the game and strategy are drawn from. We establish a grain of truth property for this case by showing that Thompson sampling versions of the Bayes-optimal strategies are in the strategy class, so are assigned non-zero probability. This allows us to use single-agent asymptotic optimality results for Thompson sampling in an unknown environment to prove convergence to a $\eps$-Nash equilibrium:
\begin{quote}
    \emph{There is a class of limit-computable strategies satisfying the grain of truth property with respect to any computable game, and it includes limit-computable strategies that converge to a $\varepsilon$-Nash equilibrium even when the game is unknown.}
\end{quote}
A brief outline of the paper including dependencies between major theorems appears below.

\paragraph{Contributions.}
We solve the long-standing grain of truth problem by introducing a class of reflective-oracle computable strategies. This allows us to establish convergence of Bayesian players to $\varepsilon$-Nash equilibrium in known repeated games, followed by convergence for Thompson sampling strategies on unknown games. The rigor and elegance of proofs are improved over the conference version of this paper \cite{leike_formal_2016} by extending reflective oracles to non-binary alphabets with ``types'' for distinct action and percept spaces. We also include a novel application of the machinery developed for the grain of truth problem to answer a question posed in \cite{catt_self-predictive_2023}, constructing a self-predictive agent with consistent beliefs about its own policy and suggesting a direction for further research. All results are shown to be limit-computable.   

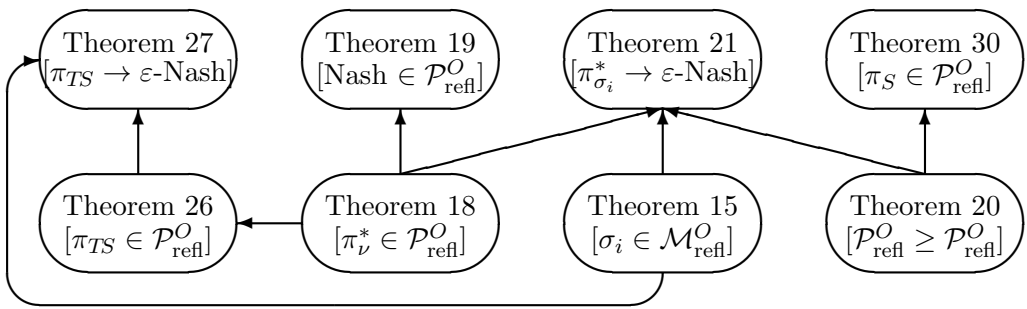
\begin{center}
\unitlength=2.4ex
\linethickness{0.4pt}
\begin{picture}(31,8)
\thicklines\small
\put(3,1.5){\oval(6,3)[cc]\makebox(0,0)[cb]{\raisebox{3pt}{Theorem~\ref{thm:est_pi_TS_comp}}}\makebox(0,0)[ct]{[$\pi_\TS \in \cP\rOi$]}}
\put(11,1.5){\oval(6,3)[cc]\makebox(0,0)[cb]{\raisebox{3pt}{Theorem~\ref{thm:O_gen_optimal_strategy}}}\makebox(0,0)[ct]{[$\pi^*_\nu \in \cP\rOi$]}}
\put(19,1.5){\oval(6,3)[cc]\makebox(0,0)[cb]{\raisebox{3pt}{Theorem~\ref{sigma_i_refl}}}\makebox(0,0)[ct]{[$\sigma_i \in \cM\rOi$]}}
\put(27,1.5){\oval(6,3)[cc]\makebox(0,0)[cb]{\raisebox{3pt}{Theorem~\ref{thm:dominant_zeta}}}\makebox(0,0)[ct]{[$\cP\rOi \geq \cP\rOi$]}}
\put(3,6.5){\oval(6,3)[cc]\makebox(0,0)[cb]{\raisebox{3pt}{Theorem~\ref{thm:limit_comp_nash_convergence}}}\makebox(0,0)[ct]{[$\pi_\TS \rightarrow \varepsilon$-Nash]}}
\put(11,6.5){\oval(6,3)[cc]\makebox(0,0)[cb]{\raisebox{3pt}{Theorem~\ref{thm:refl_oracle_comp_nash}}}\makebox(0,0)[ct]{[Nash $\in \cP\rOi$]}}
\put(19,6.5){\oval(6,3)[cc]\makebox(0,0)[cb]{\raisebox{3pt}{Theorem~\ref{thm:kalai_lehrer_convergence}}}\makebox(0,0)[ct]{[$\pi^*_{\sigma_i} \rightarrow \eps$-Nash]}}
\put(27,6.5){\oval(6,3)[cc]\makebox(0,0)[cb]{\raisebox{3pt}{Theorem~\ref{thm:self_aixi_comp_env}}}\makebox(0,0)[ct]{[$\pi_S \in \cP\rOi$]}}
\put(8,1.5){\vector(-1,0){2}}
\put(3,3){\vector(0,1){2}}
\put(11,3){\vector(0,1){2}}
\put(19,3){\vector(0,1){2}}
\put(27,3){\vector(0,1){2}}
\put(11,3){\vector(4,1){8}}
\put(27,3){\vector(-4,1){8}}
\put(9,0){\oval(20,2)[br]\oval(20,2)[bl]} 
\put(0,0){\oval(2,13)[tl]}\put(0,6.45){\vector(1,0){0}}
\end{picture}
\end{center}

\section{Mathematical Preliminaries}\label{sec:prelim}

\paragraph{Notation.}
$\cA^*$ is the set of finite strings $x$ over a finite set of alphabet symbols $a\in\cA$. We will use $x_t \in \cA$ to denote the $t^\text{th}$ element of such a string (indexing from 1) and $x_{1:t} = x_{\leq t}$ is the substring $x_1x_2...x_t \in \cA^*$; $x_{<t}$ and $x_{\neq t}$ are defined analogously. The string $xy$ is the concatenation of $x$ and $y$. We bars $|\cdot|$ are overloaded, representing length for strings, cardinality for sets, and absolute value for reals. While the index $t$ is reserved for ``temporal'' indexing such as elements of an ordered string or sequence, we will reserve $n$ for the number of players in a game and $1 \leq i,j \leq n$ for the indices of (the currently considered and other, respectively) players in the game. We will use $\cT$ to denote the set of probabilistic Turing machines with oracle access. We will use $\Delta$ to represent a probability simplex; for example $\Delta \cA$ is a probability distribution on $\cA$. For probability measures $\mu,\nu$ we write ``$\mu$ is absolutely continuous w.r.t. $\nu$'' as $\mu \ll \nu$. The Iverson brackets $[\![R]\!]$ are 1 when $R$ is true and 0 when $R$ is false (they cast booleans to integers). Further notation will be introduced as needed; see \cref{app:Notation} for a complete list.

We will need the following computability levels from the arithmetic hierarchy.

\begin{definition}[computability]
A function $f$ is\vspace{-1.5ex}
\begin{itemize}\parskip=0ex\parsep=0ex\itemsep=0ex
    \item \emph{(finitely) computable} (or recursive) if it is computed by some Turing machine. 
    \item \emph{estimable} if there is a computable function $\phi(x,k)$ such that $\forall k |f(x) - \phi(x,k)| < \frac{1}{k}$. That is, $f$ can be approximated to arbitrary pre-specified precision. 
    \item \emph{lower semicomputable} (l.s.c.) if there is a computable function $\phi(x,k)$ monotonically increasing in its second argument with $\lim_{k\to\infty} \phi(x,k) = f(x)$. That is, $f$ can be approximated from below.  
    \item \emph{limit-computable} (or approximable) if there is a computable function $\phi(x,k)\to f(x)$ for $k\to\infty$. That is, $f$ can be approximated to arbitrary but unknown precision. 
\end{itemize}
\end{definition}
An estimatable function is always l.s.c.: If $\phi$ estimates $f$, then $\phi'(x,k'):=\max_{k\leq k'}\{\phi(x,k)-\fr1k\}$ lower semicomputes $f(x)$.
Estimable functions are often called `computable,' but we find this term ambiguous.

\begin{definition}[semimeasure]\hfill\par
\noindent A semimeasure $\nu$ is a function $\cA^*\to\mathbb{R}^+$ satisfying $\nu(x) \geq \sum_{a \in \cA} \nu(xa)$. 
\end{definition}
For our purposes semimeasures always assign probability $1$ to the empty string $\eps$.
Because algorithms do not always halt, the objects of algorithmic probability are often semimeasures with  probability gaps arising from non-halting behavior. 
Semimeasures are designed for sequence prediction and assign a (defective) probability to observing a sequence starting with string $x$. Another viewpoint more in line with measure theory is that $x$ represents the cylinder set $\Gamma_x$ including all infinite continuations of $x$. A standard well-defined extension of semimeasures to (near-arbitrary) sets of infinite strings analogous to Carathéodory's extension theorem for measures is given in \cite{wyeth_hutter_semimeasures_2025}. 

 \section{The Grain of Truth Problem} \label{sec:GoT}

\paragraph{Problem statement.}
The grain of truth problem concerns a set of $n$ players engaged in a multi-player game $\sigma$ in some (countable) class of games $\cG$. Each player believes that the other players' strategies are drawn independently from a (countable) strategy (=\emph{policy}) class $\cP$. Then it is natural to ask whether under some choices of prior over $\cP$, a Bayesian optimal strategy for each player is itself in $\cP$. In other words, we are seeking conditions that make the players' beliefs about each other ``consistent'' and subjectively optimal. We can also view this as a single agent interacting with an environment $\nu$ which consists of game $\sigma$ combined with the remaining agents (in contrast to games, environments have only one ``player''). If and only if $(\cG,\cP)$ contain a grain of truth in the above sense, the condition that the true environment is $\mu\in\cM$ is satisfied.
Consider a multi-agent setup where $n$ agents interact with each other via a common environment in rounds.
In game-theoretic parlance, player $i$ follows some mixed strategy $\pi_i$ from some class of strategies $\cP$.
The extensive-form ``game'' $\sigma$ they are ``playing'' also includes observations and rewards to the agents, 
where the utility is the expected discounted reward sum. 
Conventionally a repeated known game and Nash equilibria are considered.
We deal with this case only as a stepping stone to our much more general setting:
Our main setting and results consider one long extensive-form game
which is only known to belong to a countable class of games $\cG$,
and essentially no structural assumptions are made on $\cG$.
We also do not assume that players play Nash equilibria.
Player $i$ only assumes that the others' strategies $\pi_j\in\cP$ and that the game $\sigma\in\cG$\footnote{In Section~\ref{sec:unknown_games} the players' belief distributions can be made in a sense even more general, taking a mixture over a class of subjective environments that includes any combination of a game from $\cG$ and opponent strategies from $\cP$, but need not explicitly draw a distinction between an opponent and any other part of the environment.}.
From player $i$'s perspective, he interacts with an environment $\sigma_i$ that consists
of game $\sigma\in\cG$ and a (deleted) strategy profile $\pi_{\neq i}:=(\pi_1,...,\pi_{i-1},\pi_{i+1},...,\pi_n)\in\cP^{n-1}$ of the other players $j\neq i$.
Player $i$ does not know $\sigma$ and does not assume that $\pi_{\neq i}$ is a Nash strategy, 
so needs to infer both from the experienced interaction history $h_{<t}$. 
This setup is more realistic in that it allows to model agents that learn from experience and do not assume any particular strategy (e.g.\ Nash) of their ``opponents'' beyond being in $\cP$, and do not need to know the game they are playing.
This is the multi-agent version of the optimal history-based reinforcement learning agent AIXI \cite{Hutter:04uaibook,Hutter:24uaibook2}, and it can model very general problems. For instance, player $i$ may face sub-optimal, colluding, or cooperative opponents.

\paragraph{Strategies.}
A player's strategy is a mapping from a string of previous moves in his action set $\cA$ and previous percepts from his information set $\cE$ to a (conditional) probability distribution for his next action $\Delta \cA$, with signature $\pi : (\cA \times \cE)^* \rightarrow \Delta \cA$. For instance $\cE$ may be a list of the other players' moves. We can equivalently model a strategy as a family of measures over action sequences for each percept sequence in a player's information set (by taking the product of the conditional probabilities). The probabilities of the first $t$ actions can only depend on observations available before time $t$. This is sometimes referred to as a ``chronological contextual'' measure \cite{Hutter:24uaibook2}. Typically, we will show that our grain of truth classes $\cP$ contain a ``dominant'' mixture policy $\zeta\in\cP$ such that $\forall \pi \in \cP$, $\exists c \in \mathbb{R^+}$ satisfying $\zeta \geq c\pi$. In that case, each player may express their belief that the others choose some strategy in $\cP$ by modeling their strategies as $\zeta$ (this is essentially Kuhn's theorem \cite{aumann_mixed_1964}; see also \cite{alexander_private_2023}).
In the case that each player has a different action set it is necessary to generalize this problem by indexing the policy classes by player as $(\cP_i)_{1 \leq i \leq n}$.

\paragraph{Note on computability.}
The vast majority of strategies do not satisfy any reasonable computability standards. As an (informal) example, consider $\cA = \{0,1\}$ and an empty percept space. Then any deterministic (sometimes \emph{pure}) strategy can be identified with an infinite binary sequence, which encodes a real number in the unit interval. It is a well-known result of recursion theory that almost all reals are not computable under any natural formalization (e.g. the type two theory of effectivity). We introduce relevant computability properties in \cref{sec:refl_oracles} and apply them to strategies in \cref{sec:multiagents}.

\paragraph{Multi-player games.}
Formally, a multi-player game has signature $\sigma : (\cA^n \times \cE^n)^* \times \cA^n \rightarrow \Delta \cE^n$. Given a sequence of action vectors $a_t = (a_t^1, a_t^2, ..., a_t^n) \in \cA^n$, a game $\sigma$ assigns a probability to the sequence of percept vectors $e_t = (e_t^1, e_t^2, ..., e_t^n) = (o_t^1r_t^1, o_t^2r_t^2, ..., o_t^nr_t^n) \in \cE^n$ including observations $o_t^i$ and rewards $r_t^i \in [0,1]$ (by taking the product of conditionals). This is written $\sigma(e_{\leq t}||a_{\leq t})$. The probabilities of percepts received at time $t$ can only depend on the actions taken up to and including time $t$.

\paragraph{Action/percept encodings.}
The settings we will discuss restrict players and games to be represented by probabilistic Turing machines, so that they accept the interaction history on an input tape and stochastically produce an action or observation on the output tape. The symbolic representation of the history on these tapes can become important. We require it to be uniquely decodeable into actions and percepts (for instance, by devoting a consistent number of symbols to each action and percept). We will denote the representation of any object $s$ as $\langle s \rangle$ and assume that $\langle a_1^ie_1^i...a_t^ie_t^i \rangle = \langle a_1^i \rangle \langle e_1^i \rangle ... \langle a_t^i \rangle \langle e_t^i \rangle.$ The simplest ``highly granular'' representation is to use a unique symbol for each $e_t^i \in \cE$\footnote{A perception can be encoded as a single symbol as long as the observation and reward can be computably extracted.} and a unique symbol for each $a_t^i \in \cA$. Conceptually these symbol sets should be disjoint, but it is always possible to determine which is which by indexical position. 
By the compositionality properties of Turing machines, if a player's opponents have computable strategies and the game is computable, we can construct a Turing machine to simulate both the opponents and the game from his perspective; this forms a computable \emph{subjective environment}. 

A Bayesian learner chooses Bayes-optimal actions w.r.t.\ the Bayes-mixture over $(\sigma,\pi_{\neq i})$
w.r.t.\ some prior $w$ over $\cG\times\cP^{n-1}$. If there is such a Bayes-optimal strategy in $\cP$, we say that $\cG, \cP$ satisfies the grain of truth property with respect to this choice of prior.
There are general theorems which show that (a Thompson-sampling version of) the Bayes-optimal
strategy $\pi_i^*$ is asymptotically optimal in the sense that it converges 
to the optimal (informed) agent who knows the environment (here $\sigma,\pi_{\neq i}$).
This single-agent view is asymmetric in that it singles out one (Bayes-optimal) agent $i$ against $n-1$ agents from some class $\cP$.
A symmetric treatment requires agent $i$ to consider the possibility that the other agents $j\neq i$ are also Bayes-optimal agents $\pi_j^*$.
This consideration is formally satisfied iff (the Thompson sampling version of) $\pi_j^*$ is itself in $\cP$ with a non-zero prior $w(\pi_j^*)>0$.
We say that $(\cP,\cG,w)$ contains a \emph{grain of truth} if this condition is satisfied.
The grain of truth problem is the question of whether there exist interesting classes $(\cP,\cG)$ that contain a grain of truth.

As a special case, we construct an interesting policy class satisfying the convergence conditions of \cite{kalai_rational_1993} in known infinitely repeated games. For these purposes, we prove a weak form of the grain of truth property, formally:

\begin{definition}[Grain of truth property]\label{def:grain_of_truth}
    Given a class of policies=strategies $\cP$ and a class of games $\cG$, consider a vector of policies $\pi = (\pi_1,...,\pi_n)\in\cP^n$.
    Hold out one player $i$ and construct the environment $\sigma_i^\pi$ it is interacting with (the game $\sigma$ together with the other players' $j\neq i$ strategies). Let player $i$'s beliefs about his environment be described by the Bayesian mixture environment $\xi_i$.  Let $\pi_{\xi_i}^*$ be some Bayes-optimal strategy for $\xi_i$. We say that $(\cP,\cG)$ contains a grain of truth iff $\forall i,\ \forall \pi_{\neq i} \in \cP^{n-1},\ \forall \sigma \in \cG,\ \sigma^\pi_i \ll \xi_i$ and $\pi^*_{\xi_i} \in \cP$. 
\end{definition}
This is the property we will need to show convergence to $\epsilon$-Nash equilibrium in a known game. The condition $\sigma^pi_i \ll \xi_i$ can be satisfied by choosing $\xi_i$ as an explicit Bayesian mixture over $\cG \times \cP^{n-1}$ with weights $w_i$, in which case it makes sense to say that $(\cP,\cG,w)$ satisfies the grain of truth property, but this is not required by \cref{def:grain_of_truth}. For unknown games $(\cG \neq \{\sigma\})$ we need a stronger property for convergence:
\begin{definition}[Strong grain of truth property]\label{def:strong_grain_of_truth}
    Given a class of policies=strategies $\cP$ and a class of games $\cG$, let player $i$'s beliefs about his environment be described by the Bayesian mixture $\xi_i$ with weights $w_i(\nu) > 0$ for $\nu \in \cM$.  Let $\pi_{\xi_i}^*$ be some Bayes-optimal strategy for $\xi_i$. We say that $(\cP,\cG)$ and specifically $(\cP,\cG,w)$ contains a strong grain of truth iff $\forall i,\ \forall \pi_{\neq i} \in \cP^{n-1},\ \forall \sigma \in \cG,\ \sigma^\pi_i \in \cM$ and $\pi^*_{\xi_i} \in \cP$. 
\end{definition}
In fact, we need the ``Thompson sampling version'' of \cref{def:strong_grain_of_truth}, obtained by replacing $\pi^*_{\xi_i}$ by the Thompson sampling strategy $\pi_\TS$ in the final condition. It is clear that the strong grain of truth property is in fact stronger than the grain of truth property; it implies not only absolute continuity but even bounded Radon-Nikodym derivative $d\sigma^\pi_i/d\xi_i \leq w_i(\sigma^\pi_i)^{-1}$. 

\paragraph{History of the grain of truth problem.}
Progress towards discovering rich strategy and game classes satisfying the grain of truth property has been slow. After \cite{kalai_rational_1993} introduced the grain of truth property in the context of infinitely repeated games and independent strategies, along with a simple prisoner's dilemma example, many impossibility results were proven (\cite{nachbar_prediction_1997,nachbar_beliefs_2005,foster_impossibility_2001}). Much later, \cite{fallenstein_reflective_2015} introduced reflective oracles, laying the groundwork for a solution (but only considering stage games and not the grain of truth property). This work was extended to sequential decision theory by \cite{aixi_reflective_2015}, but \cite{leike_formal_2016} (the conference version of this paper) was the first to solve the grain of truth problem. However, Leike et al. focused on convergence in unknown games which required Thompson sampling strategies instead of the Bayesian strategies of \cite{kalai_rational_1993}. As a result they did not fully correctly formulate or prove a solution to Kalai and Lehrer's problem. We provide such a solution along with more elegant and complete proofs of many of Leike et al.'s other results.

\begin{table}[t]
\begin{center}
\begin{tabular}{p{0.445\columnwidth}p{0.445\columnwidth}}
\toprule
Reinforcement learning & Game theory \\
\midrule
stochastic policy & mixed strategy \\
deterministic policy & pure strategy \\
agent & player \\
multi-agent environment & infinite extensive-form game \\
reward & payoff/utility \\
(finite) history & history \\
infinite history & path of play \\
\bottomrule
\end{tabular}
\end{center}
\caption{Terminology dictionary between reinforcement learning and game theory from \cite{leike_formal_2016}.}
\label{tab:rl-game-theory-translation}
\end{table}

\paragraph{Examples.}
Any Nash equilibrium of a game $\sigma$ with strategy $\pi_i$ for each player $i$ is a trivial ``solution'' to the grain of truth problem with $\cG = \{ \sigma \}$, $\mathcal{P}_i = \{ \pi_i \}$, but this is not very interesting for our purposes because it does not necessarily model learning. 

A basic but non-trivial example is discussed \cite{kalai_rational_1993}; consider an infinitely repeated prisoner's dilemma. In every time step the payoff matrix is as follows,
where C means cooperate and D means defect.
\begin{center}
\begin{tabular}{l|cc}
  & C        & D \\
\hline
C & 3/4, 3/4 & 0, 1 \\
D & 1, 0     & 1/4, 1/4
\end{tabular}
\end{center}

Let the strategy class be $\cP = \{g_t\}_{t \in \mathbb{N}\cup\{\infty\}}$ where $g_t$ is a grim trigger strategy that punishes defection by defecting indefinitely, but by default cooperates until time $t$ and defects afterwards. It is fairly easy to see that regardless of a player's prior belief $w_t$ in each $g_t$, once he or his opponent has defected, any strategy in $\cP$ continues to defect indefinitely, so he expects his opponent to certainly defect. The Bayes optimal strategy, which is strictly dominant in the game theoretic (not our measure-theoretic) sense, is for him to defect from that point on, which is itself a grim trigger strategy. Depending on his priors $w_t$, he may earlier expect his opponent to very likely defect at time $t_d < \infty$ despite continued cooperation for $t \leq t_d$, in which case his optimal policy may be $g_{t_d-1}$. Not only does this $(\cP,\cG) = (\cP,\{\sigma\})$ satisfy the strong grain of truth property, $(\cP,\cG,w)$ satisfies the strong grain of truth property for any choice of $w_i$ supported on $\{\sigma^\pi | \pi_{\neq i} \in \cP^{n-1}\}$. The catch is that $\cG$ is only a single (known) environment. 

For a much larger (non)example, the class containing all strategies naively appears to satisfy the grain of truth property, but in any nontrivial infinite game it is not countable and certainly has no dominant strategy, so it is usually not possible to define a useful prior over this class.

\paragraph{} 
Using the terminology introduced above, we can rephrase our main result as follows:

\begin{theorem}[limit-computable convergence to equilibrium]
    There are limit-computable (Thompson sampling) strategies  $\pi_1, ..., \pi_n$ such that for any computable multi-player game $\sigma$ and for all $\eps > 0$ and all $i \in \{ 1, ..., n \}$ the $\sigma^{\pi_{1:n}}$-probability that the policy $\pi_i$ is an $\eps$-best response converges to 1 as $t\to\infty$.
\end{theorem}

\section{Reflective Oracles}\label{sec:refl_oracles}

Rational players can face an infinite regress in which each mutually reasons about the other's reasoning. For instance, if each player's strategy is computed by a commonly known Turing machine, it would seem to be rational to run the other players' machines to predict their behavior and choose the utility maximizing response. When the other players' Turing machines halt, this is computable. Unfortunately, if all players attempt this strategy there is mutual recursion as player 1 simulates player 2 simulating player 1 ad infinitum\footnote{An amusing example of this behavior appears in William Goldman's ``The Princess Bride," when Vizzini attempts to determine which of two cups Westley poisoned by speculating about what Westley will think Vizzini thinks about Westley. This is isomorphic to the game of matching pennies described in Section~\ref{sec:convergence_for_Bayesian}.}.\ 
Classically this interdependence of optimal strategies is resolved by assuming players will choose a Nash equilibrium, but this is not a Bayesian optimality notion because it is not clear why players should learn to play any particular equilibrium strategy. To model the (subjective) uncertainty of Bayesian players, as well as intentionally randomized behavior strategies, we will use probabilistic Turing machines (pTM's). Formally, a probabilistic Turing Machine (pTM) is a Turing Machine with access to both the ordinary input, output, and work tapes and an additional infinite tape initialized with random bits. To cut through the infinite regress and allow players to consistently reason about each others' strategies, we will allow all pTM's access to the same ``reflective'' oracle. In this section we show how pTM's model probability distributions (such as behavior strategies) and introduce reflective oracles. Combining these two ideas to construct pTM's with reflective oracle access, we establish some basic computability properties that lay the foundations for the rest of the paper. 

\paragraph{Sampling from a pTM.}
Probabilistic Turing machines are interpreted as computing conditional probabilities. For a pTM $T$, we define $\lambda_T(\alpha|x)$ to be the probability that on input $x$, $T$ produces the symbol $\alpha$ (and nothing else). Then we can define a semimeasure
\begin{equation}
    \lambda_T(x) = \prod_{t=1}^{|x|} \lambda_T(x_t | x_{<t})
\end{equation}
This is not in general a proper probability measure because there is some chance that the pTM does not halt or produces an invalid output. Later, we will find an interesting way to use reflective oracles to complete these semimeasures to probability measures.  

\paragraph{Oracle access.}
 Oracle access means that the pTM's can write a query on an oracle tape and enter a special state that queries the oracle, with the next transition depending on the oracle's output. We show in Appendix~\ref{sec:lscm_vs_pTM}, Theorem~\ref{thm:lscsm_from_pTM} that for any pTM $T$, $\lambda_T$ has l.s.c.\ conditionals. It is possible to invert this construction, and the other direction (Theorem~\ref{thm:pTM4lscsm}) works even when machines have access to oracles. We define $\lambda_T^O(\alpha|x)$ as the probability that the oracle pTM $T$ with access to $O$ returns $\alpha$ on input $x$. The semimeasure $\lambda_T^O$ is defined analogously to before. We will use the symbol $\nu$ to represent arbitrary semimeasures.

\begin{theorem}[l.s.c.\ semimeasures vs pTM semimeasures]\label{thm:lscsm_vs_pTM}\hfil\par
    A semimeasure $\nu$ has l.s.c.\ conditionals \emph{iff} there exists a pTM $T$ such that $\nu=\lambda_T$
\end{theorem}
\begin{proof} \emph{sketch (detailed proof in Appendix~\ref{sec:lscm_vs_pTM})}:\\
($\Leftarrow$) That the conditionals of $\lambda_T$ are l.s.c.\ is rather straight-forward from their construction.\\
($\Rightarrow$) Let $\phi_\alpha(x,k)$ be computable and monotone increasing in $k$ converging to $\nu(\alpha|x)$ for $\alpha\in\cA=\{1,...,d\}$.
Consider a pTM $T$ implementing the following procedure:
Let $\Delta_\alpha(k):=\phi_\alpha(x,k)-\phi_\alpha(x,k-1)\geq 0$ with $\phi_\alpha(x,0):=0$.
Then chop 
\begin{align*}
  \text{successive intervals}~~~ & ~I_1(1),...,~I_d(1),~I_1(2),...,~I_d(2),~I_1(3),... \\
  \text{of lengths}~~~ & \Delta_1(1),...,\Delta_d(1),\Delta_1(2),...,\Delta_d(2),\Delta_1(3),...  
\end{align*}

from interval $[0;1)$. All-together these intervals cover $[0;\sum_\alpha \nu(\alpha|x))\subseteq[0;1)$.
Let $\omega_{1:\infty}$ be uniform random bits.
Let $T$ output $\alpha$ if $\exists k:[0.\omega_{1:k};0.\omega_{1:k}+2^{-k}) \subseteq \bigcup_{k'=1}^\infty I_\alpha(k')$.
For $0.\omega < \sum_\alpha \nu(\alpha|x)$, the condition can be tested effectively by running through $k=1,2,3,...$ while only finitely many $k'$ need to be checked.
The procedure terminates in finite time, since the interval on the l.h.s.\ shrinks to a point ($0.\omega$) for $k\to\infty$,
hence eventually is contained in some $I_\alpha(k')$.
This procedure outputs $\alpha$ with probability $\nu(\alpha|x)$, 
since 
$$
  \textstyle P[0.\omega\in\bigcup_{k'=1}^\infty I_\alpha(k')] 
  ~=~ |\bigcup_{k'=1}^\infty I_\alpha(k')|
  ~=~ \lim_{k\to\infty} \sum_{k'=1}^k \Delta_\alpha(k')
  ~=~ \lim_{k\to\infty}\phi_\alpha(x,k) ~=~ \nu(\alpha|x)
$$

For $0.\omega \geq \sum_\alpha \nu(\alpha|x)$ no $k$ is found, and $T$ runs forever with no output (which is fine).
\end{proof}
%

\paragraph{$O$-sampled conditionals.}
Without $O$ access, \cref{thm:lscsm_vs_pTM} shows that a semimeasure has l.s.c.\ conditionals iff it is sampled by a pTM. However, because the oracle $O$ may be probabilistic, it is not clear that the $\Leftarrow$ direction still holds with oracle access. Therefore, we will avoid Leike et al.'s \cite{leike_formal_2016} potentially misleading terminology ``l.s.c.\ with oracle access'' for these semimeasures. Instead, we will say that a semimeasure $\mu$ has $O$-sampled conditionals (or ``is $O$-sampled'' for brevity) if there is a pTM $T$ such that $\mu(\alpha|x) = \lambda_T^O(\alpha|x)$ for $\alpha \in \cA$.

\paragraph{$O$-estimable conditionals.}
Following the convention set by ``$O$-sampled'' conditionals, we will use the term ``$O$-estimable'' conditionals to refer to semimeasures that have conditional probabilities estimable with $O$ access. When it is clear from context we will drop the word ``conditionals.''

\paragraph{Formalizing oracles.}
For our purposes, oracles always answer queries with 0 or 1 (which can be interpreted as false or true). Because they are allowed to (independently) randomize their answers on queries, an oracle's behavior is specified by its probability of answering 1. This means we can treat oracles as functions to the unit interval. 
\begin{definition}[reflective oracle]\label{def:refl_oracle}
An oracle $O : \cT \times \cA^* \times (\mathbb{Q} \cap [0,1]) \times \cA \rightarrow [0,1]$ is called reflective iff for each pTM $T$ and string $x \in \Sigma^*$, $\exists \{ q_\alpha \}_{\alpha \in \cA}$ satisfying the following properties:
\begin{equation} \label{eq:q_on_simplex}
    \sum_{\alpha \in \cA} q_\alpha = 1
\end{equation}
And for all $\alpha \in \cA$ and $p\in\mathbb{Q}$,\vspace{-2ex}
\[
\lambda_T^O(\alpha|x) \leq q_\alpha \leq 1 - \sum_{\beta \neq \alpha} \lambda_T^O(\beta|x)
\]
\[
O_\alpha(T,x,p) = 1 ~~~\text{for}~~~ p < q_\alpha
\]
\[
O_\alpha(T,x,p) = 0 ~~~\text{for}~~~ p > q_\alpha
\]    
\end{definition}
This is the same as \cref{def:reflective_oracle} restated to take advantage of our $\lambda^O_T$ notation. We will often abbreviate ``reflective oracle'' as \rO.

We will reserve the notation $O_\alpha(T,x,p) \rightarrow 0$ (respectively 1) for the event that reflective oracle $O_\alpha$ called on the query $(T,x,p)$ yields response 0 (respectively 1). This occurs with probability $O_\alpha(T,x,p)$ by definition, so ``calling'' $O_\alpha(T,x,p)$ is equivalent to invoking $\flip(O_\alpha(T,x,p))$ where $\flip(p)$ is a function that returns 1 with probability $p$ and 0 with probability $1-p$.  

Because of equation \cref{eq:q_on_simplex}, $q_\alpha$ can be viewed as a conditional probability assignment for each symbol $\alpha\in\cA$. When $\lambda^O_T$ is a measure, the query $(T,x,p)$ can be viewed as asking the question ``is (case 0) $p > \lambda^O_T(\alpha|x)$ or (case 1) $p < \lambda^O_T(\alpha|x)$?" then $O$'s answers are consistent with $q_\alpha = \lambda^O_T$; always 1 when $p < q_\alpha$ and 0 when $p > q_\alpha$, but allowed to randomize when $p = q_\alpha$ exactly.
This randomization means that $q_\alpha$ cannot be determined exactly and avoids diagonalization.
When $\lambda^O_T$ is only a defective semimeasure, its conditionals do not sum to 1 so cannot satisfy equation \cref{eq:q_on_simplex}, which means that $q_\alpha \neq \lambda^O_T(\alpha|x)$; however the definition requires at least $q_\alpha \geq \lambda^O_T(\alpha|x)$. This means that $O$ ``redistributes" the non-halting probability mass of $T^O$ and completes $\lambda^O_T$ to a measure. The requirement $q_\alpha \leq 1 - \sum_{\beta \neq \alpha} \lambda_T^O(\beta|x)$ is actually redundant because it follows from $q_\alpha \geq \lambda^O_T(\alpha|x)$ and $\sum_{\alpha \in \cA} q_\alpha = 1$. The existence of reflective oracles on non-binary alphabets is proven in \cref{sec:non_binary_refl_oracle_existence}.

Fallenstein et.\ al.\ originally defined reflective oracles for a binary alphabet in an analogous way \cite{fallenstein_reflective_2015}. Leike \cite{leike_formal_2016} used a more general definition which allowed $O$ to randomize arbitrarily in the entire range $\lambda_T^O(\alpha|x)$ to $1 - \sum_{\beta \neq \alpha} \lambda_T^O(\beta|x)$. Leike's general definition slightly simplifies his proof of limit-computability, but is only specified for the binary case and is not easy to directly extend to the non-binary case. When it is necessary to distinguish between the two cases we will call reflective oracles satisfying Fallenstein et.\ al.'s and our stricter definition ``step reflective oracles.'' 

Let $\bar\lambda_T^O$ be the completion of $\lambda_T^O$ by a reflective oracle $O$, with $\bar\lambda_T^O(\alpha|x) = q^O_{\alpha,T,x}$, where $q^O_{\alpha, T,x}=q_\alpha$ as defined in \cref{def:refl_oracle}. This is a properly normalized probability measure by equation \cref{eq:q_on_simplex}. Note that $\bar\lambda$ is a function of $O$ and $T$ (producing a measure). The completion (bar) is not applied as an operator to $\lambda_T^O$, because many different pTM's may produce the same semimeasure which can be completed in different ways. For example when $T$ does not make oracle calls and $\lambda_T^O = \lambda_T$ is defective (say, 0 everywhere) it can be completed arbitrarily with appropriate choice of \rO\ as mentioned in Appendix~\ref{sec:non_binary_refl_oracle_existence}.
%
\begin{theorem}[properties of $\bar\lambda_T^O$]\label{thm:bar_lambda_est_properties}
    For any pTM $T$, $\bar\lambda_T^O$ is an $O$-estimable probability measure. In particular, there is an oracle pTM $B_T$ 
    \ estimating $\bar\lambda_T^O$ that is computably constructable from $T$.
\end{theorem}
\begin{proof}
    Given any reflective oracle $O$, for each pTM $T$, and string $x$ there are particular (clearly unique) values $q_{\alpha, T,x}^O$ satisfying the above requirements for $q_\alpha$. There is a pTM $B_T$ with $O$ access that conducts a binary search for $q_{\alpha, T, x}^O$ by using queries to $O$ to determine whether each $p$ is above or below $q_{\alpha, T,x}^O$. This process may behave stochastically if $q_{\alpha, T,x}^O$ itself is ever a query, but the limit is always correct. Since the range of possible values for $q_{\alpha, T, x}^O$ halves with each query, it is $O$-estimable. 
\end{proof}
Notably, our procedure for estimating $\bar{\lambda}^O_T$ does not involve simulating $T$ as in the procedure to l.s.c.\ $\lambda_T$ (which does not work with oracle access), but only uses the description of $T$ to run the binary search $B_T$. This is related to $\lambda^O_T$ \emph{only} because reflectivity of $O$ leads to $\bar{\lambda}^O_T \geq \lambda^O_T$.\\

\paragraph{Reflective Oracles and Diagonalization.}
\label{ex:diagonalization-for-reflective-oracles}

The following example is from \cite{fallenstein_reflective_2015}. Let $T \in \mathcal{T}$ be a probabilistic Turing machine with a two symbol output alphabet $\cA = \{\alpha, \beta\}$ that
outputs $\beta$ if $O_\alpha(T, \epsilon, 1/2) \rightarrow 1$ and $\alpha$ if $O_\alpha(T, \epsilon, 1/2) \rightarrow 0$. 
($T$ can know its own source code by quining~\cite[Thm.~27]{Kleene1952-KLEITM}).
In other words, $T$ queries the oracle about whether it is more likely
to output $\alpha$ or not, and then does whichever the oracle says is less likely.
In this case we can use an oracle $O_\alpha(T, \epsilon, 1/2) := 1/2$
(answer $0$ or $1$ with equal probability),
which implies $\lambda_T^O(\alpha | \epsilon) = \lambda_T^O(\beta | \epsilon) = 1/2$,
so the conditions of Definition~\ref{def:refl_oracle} are satisfied.
In fact, for this machine $T$ we must have
$O_\alpha(T, \epsilon, 1/2) = 1/2$ for all reflective oracles $O$.

%
\begin{theorem}[pTM for $\bar{\lambda}^O_T$] \label{thm:bar_lambda_O_sampled}
    For any reflective oracle $O$, all $O$-estimable semimeasures are $O$-sampled. In particular, for any pTM $T$, $\bar\lambda_T^O$ is $O$-sampled. 
\end{theorem}
\begin{proof}
    Any estimable function is also l.s.c.\ since the lower bound of the estimate can be used as the approximation from below. This means that the sampling algorithm (\cref{alg:sample}) can be used to sample from any $O$-estimable probability measure (in this case halting with probability 1). Therefore, $\bar\lambda_T^O$ is also $O$-sampled; assuming that pTM $S$ implements the sampling algorithm and accepts its argument $\phi_\alpha$ in the form of a pTM encoding, and that the binary search pTM $B_T$ returns the low end of its interval estimates, $\bar\lambda_T^O = \lambda_{S(B_T, \cdot)}^O$.
\end{proof}
More succinctly, ``$O$-estimable conditionals'' implies ``$O$-sampled conditionals". The converse does not hold because a semimeasure is not necessarily equal to its completion, but the converse does hold for probability measures. See Appendix \ref{sec:general_refl_oracle_comp} for proofs using Leike et al.'s definition.

\begin{lemma}[all estimable measures $O$-sampled]\label{lemma:est_O_sampled}
    For any (joint) estimable measure $\nu$ there exists one pTM $T$ such that $\nu = \bar{\lambda}^O_T$ regardless of the choice of reflective oracle $O$. 
\begin{proof}
    This theorem is a stronger version of \cref{thm:bar_lambda_O_sampled} that applies when $\nu$ is estimable without oracle access, but requires only $\nu(x)$ to be estimable, not the conditional $\nu(\cdot|x)$. Let $T$ sample its output $\alpha$ from an estimate of $\nu(x\alpha)/\nu(x)$. This is estimable except when $\nu(x)=0$, so for $\nu(x) \neq 0$, $\lambda_T^O(\alpha|x) = \nu(\alpha|x)$. When $\nu(x) = 0$, $T$ may never halt and $O$ completes $\lambda^O_T(\cdot|x)$ in some arbitrary way. This only affects conditionals for strings that already have probability 0 so the product defining $\bar{\lambda}^O_T$ still assigns all continuations probability 0 and $\bar{\lambda}^O_T = \nu$. 
 \end{proof}
\end{lemma}

The original construction of $O$ in \cite{fallenstein_reflective_2015} involved a non-constructive fixed-point argument implicitly invoking a continuous ``hierarchy'' of oracles. It looked like $\bar{\lambda}^O_T$ may not even be expressible within the arithmetic hierarchy. Surprisingly, we can choose $O$ so that $\bar{\lambda}^O_T$ is limit-computable (without requiring $O$ access, instead limit-computing $O$) by the following result:
%
\begin{theorem}[a limit-computable reflective oracle {\cite[Thm.6]{leike_formal_2016}}] \label{thm:O_limit_comp}\hfill\par 
    \noindent There is a limit-computable (binary alphabet) reflective oracle.
\end{theorem}
We show in \cref{thm:non_binary_limit_comp_refl_oracle} that there are also limit-computable non-binary alphabet reflective oracles.

\section{Multi-Player Games}\label{sec:multiagents}
Now we are ready to formally define multi-player games and strategies. We will show how multi-player games give rise to a subjective environment for each player. We refer to Bayes-optimal strategies in a subjective environment as Bayesian strategies in the associated multi-player game. Next we use the computability results established in \cref{sec:refl_oracles} to introduce reflective-oracle computable strategies and show they are effectively enumerable, which prepares us to describe players' beliefs with Bayesian mixture strategies. Together these results allow us to describe Bayesian players who believe that strategies are \rOc. 
\subsection{Definitions}\label{multiagent_defs}

We define multi-player games following \cite[Sec.7.3]{leike_formal_2016}:
In a multi-player game, $n$ players take sequential actions from $\cA$ independently and in parallel. In step $t$, the game receives a vector of actions $a_t\in\cA^n$ where action $a_t^i\in\cA$ corresponds to player $i$. The history of actions including $a_t$ determines a stochastic ``move by nature'' containing an $n$ percept vector $e_t\in\cE^n$ where player $i$ only sees $e_t^i\in\cE$. Players can only see their own actions (though of course the percept might include the other players' actions in some games). As before, $e_t^i = o_t^ir_t^i$ where $r_t^i \in [0,1]$ is a reward.

\begin{figure}[t]
\begin{center}
\begin{tikzpicture}[scale=0.25,line width=1pt] 
\draw (0,16) -- (10,16) -- (10,20) -- (0,20) -- (0,16);
\node at (5,18) {agent $\pi_1$};

\draw (0,10) -- (10,10) -- (10,14) -- (0,14) -- (0,10);
\node at (5,12) {agent $\pi_2$};

\node at (5,7) {\vdots};

\draw (0,0) -- (10,0) -- (10,4) -- (0,4) -- (0,0);
\node at (5,2) {agent $\pi_n$};

\draw (18,0) -- (30,0) -- (30,20) -- (18,20) -- (18,0);
\node at (24,10) {\begin{minipage}{22mm}
multi-player \\ game $\sigma$
\end{minipage}};

\draw[->] (10,18.5) to node[above] {$a_t^1$} (18,18.5);
\draw[<-] (10,17.5) to node[below] {$e_t^1$} (18,17.5);
\draw[->] (10,12.5) to node[above] {$a_t^2$} (18,12.5);
\draw[<-] (10,11.5) to node[below] {$e_t^2$} (18,11.5);
\draw[->] (10,2.5) to node[above] {$a_t^n$} (18,2.5);
\draw[<-] (10,1.5) to node[below] {$e_t^n$} (18,1.5);
\end{tikzpicture}
\end{center}
\caption[The multi-agent model]{%
Agents $\pi_1, \ldots, \pi_n$ interacting in a multi-player game.
}
\label{fig:multi-agent-model}
\end{figure}
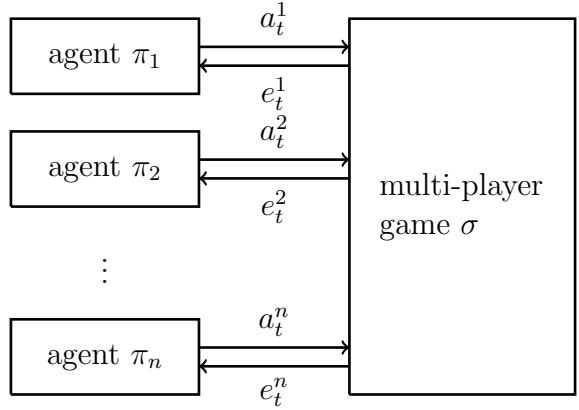

Formally,

\begin{definition}[multi-player game]\label{def:mp_game}
    A multi-player game is a function
    \[
    \sigma : (\cA^n \times \cE^n)^* \times \cA^n\to\Delta(\cE^n)
    \]
\end{definition}

The interaction of the player strategies $\pi = (\pi_1, \pi_2, ..., \pi_n)$ with the multi-player game $\sigma$ induces a history distribution $\sigma^\pi = \sigma^{\pi_{1:n}}$ where
\[
\sigma^\pi(\eps) ~:=~ 1
\]
\[
\sigma^\pi(\h_{1:t}) ~:=~ \sigma^\pi(\h_{<t}a_t) \sigma(e_t | \h_{<t}a_t)
\]
\[
\sigma^\pi(\h_{<t}a_t) ~:=~ \sigma^\pi(\h_{<t}) \prod_{i=1}^n \pi_i(a_t^i | \h_{<t}^i)
\]

Because players choose their actions simultaneously, the action distributions at time $t$ are independent conditional on the action observations history $\text{\ae}^i_{<t} := a^i_1e^i_1a^i_2e^i_2...a^i_{t-1}e^i_{t-1}$, so we take a product.
The history distribution for player $i$ is the history distribution $\sigma^\pi$ marginalized over the actions and observations of the other players:
\[
\sigma_i^{\pi}(\h_{<t}^i) ~:=~ \sum_{\h_{<t}^j, j \neq i} \sigma^{\pi_{1:n}}(\h_{<t})
\]
 
The subjective environment $\sigma_i(e_t^i | \h_{<t}^ia_t^i) = \sigma_i^{\pi}(e_t^i | \h_{<t}^ia_t^i)$ for single player/agent $i$ is actually independent of $\pi_i$ (see \cref{app:subjective_environment}), though it does depend on $\pi_j$ for $j \neq i$. Therefore we will sometimes use $\sigma^\pi_i$ to refer to the subjective environment. 

In the single agent-environment setting \cite{Hutter:04uaibook} 
$\sigma_i$ corresponds to the true environment which we write as $\mu\in\cM$ and no superscripts $i$.

\begin{definition}[environment]\label{def:environment}
    An environment $\mu$ is a chronological action-contextual measure. Equivalently, $\mu$ can be specified by its conditional probabilities $\mu(\cdot|h_{<t}a_t) \in \Delta \cE$ given every history $h_{<t} = a_1e_1...a_{t-1}e_{t-1}$ that it assigns a nonzero probability.
\end{definition}

\subsection{Strategies}\label{sec:strategies}

\begin{definition}[reflective-oracle computable strategies]
    Given a reflective oracle $O$ for action space $\cA$ and pTM's with alphabet $\Sigma = \cA \sqcup \cE$\footnote{It is acceptable for pTM's to output any symbol in $\Sigma$ as long as $O$ is indexed by $\cA$ so that conditionals are completed to $\Delta \cA$. Producing the wrong type of output is treated the same as failing to halt. Later we will explicitly introduce types for symbols, allowing $O$ to be indexed by any symbol of $\Sigma$ so that the conditionals for symbols in $\cE$ are also completed to $\Delta \cE$.}, so that $O's$ second argument is in $\Sigma^*$ and $O$ is indexed by $\cA$, we say that strategy $\pi$ is reflective-oracle computable (equivalently $O$-sampled or $O$-estimable) iff for some oracle pTM $T$, $\forall a_{<t}\in\cA^{t-1}$ and $e_{<t}\in\cE^{t-1}$ and $a\in\cA$ we have $\pi(a|\h_{<t}) = \lambda_T^O(a|\h_{<t})$. We will abbreviate reflective-oracle computable as ``\rOc''\ and refer to this class of strategies as $\cP\rOi$.
\end{definition}

\paragraph{Enumerability of $\cP\rOi$.}
Note that equivalence holds by \cref{thm:bar_lambda_est_properties} and \cref{thm:bar_lambda_O_sampled} because all strategies are assumed to be (chronological, observation contextual, proper) probability measures, because all \cref{sec:refl_oracles} theorems immediately generalize to allow input strings over $\Sigma$. The class of \rOc\ strategies is effectively enumerable as $\bar\lambda_T^O$ for $T \in (T_1, T_2, ...)$ an effective enumeration of pTMs. This enumeration contains all \rOc\ measures because oracle completion leaves probability measures unchanged. Conversely, all $\bar\lambda_T^O$ are $O$-sampled by \cref{thm:bar_lambda_O_sampled}, which means they are \rOc.

\section{Reflective-Oracle Computable Nash Equilibrium}\label{sec:refl_nash_equilibrium}

We want to construct a set of mutually expected reward maximizing strategies $\pi^*_{\sigma_1}, \pi^*_{\sigma_2}, ..., \pi^*_{\sigma_n}$ for $\sigma$. It is not obvious that this is possible because the optimal strategy for each player depends on every other players' strategy. Explicitly, $\pi^*_{\sigma_i}$ depends on $\sigma_i$ which depends on each other $\pi^*_{\sigma_j}$, which itself depends on $\sigma_j$, which depends (circularly) on $\pi^*_{\sigma_i}$. However, given an assignment of strategies to players it is certainly well-defined to discuss whether each of them is optimal given the others (i.e. a best response).
%
\begin{theorem}[subjective environment estimable] \label{sigma_i_refl}
    Given pTM's generating the multi-player game $\sigma$ and oracle pTM's generating the strategies $\pi_1, \pi_2, ..., \pi_n \in \cP\rOi$ there is an algorithm that constructs oracle pTM's estimating $\sigma^\pi$, $\sigma_i^{\pi}$, and $\sigma_i$. 
\end{theorem}
\begin{proof} Because $\sigma$ is sampled by a pTM, it is l.s.c.\ by \cref{thm:lscsm_vs_pTM}. Because it is a l.s.c.\ probability measure it is estimable. Because $\sigma^\pi$, $\sigma_i^{\pi}$, and $\sigma_i$ are defined by uniformly continuous operations on $\sigma$ and $\pi_1, ...,\pi_n$, their conditionals are also $O$-estimable (by computably constructable oracle pTM's).
\end{proof}
By an optimal strategy, we mean one that maximizes the expected sum of a player's discounted future rewards. Formally, a discount factor $\gamma_i \geq 0$ scales the reward at step $i$ to ensure the sum is finite. We will assume w.l.o.g. that the rewards are bounded so that $\sum_{i=1}^\infty \gamma_k < \infty$ ensures that $\sum_{i=t}^\infty \gamma_i r_i$ always exists. We define the discount normalization factor $\Gamma_t = \sum_{i=t}^\infty \gamma_i$, and

\begin{definition}[value function]\label{def:value_function}
    The value function of strategy $\pi$ interacting with (subjective) environment $\nu$ is \vspace{-3ex}
    \[
        V_\nu^\pi(h_{<t}) ~=~ \frac{1}{\Gamma_t} \lim_{T\to\infty} \sum_{\h_{t:T}} \sum_{i=t}^T \gamma_i r_i \prod_{j=t}^T \pi(a_j|h_{<t}\h_{<j}) \nu(e_j | h_{<t} \h_{<j} a_j) 
    \]
    which satisfies the Bellman equation
    \[
    V_\nu^\pi(h_{<t}) ~=~ \frac{1}{\Gamma_t} \sum_{a_t, e_t} \pi(a_t|h_{<t})\nu(e_t|h_{<t}a_t)( \gamma_t r_t + V^\pi_\nu(h_{<t}\text{\ae}_t))
    \]
    In a multiplayer game $\sigma$, player $i$'s value function in his subjective environment $V^{\pi_i}_{\sigma_i}$ is in game theoretic terms \cite{kalai_subjective_1993} his expected utility $U_i(\pi)$.
\end{definition}

\begin{definition}[optimal strategy $\pi^*_\nu$]\label{def:optimal_strategy}
    An optimal strategy $\pi^*_\nu$ for environment $\nu$ is a strategy in $\argmax_\pi V^\pi_\nu$, which is nonempty by \cite{lattimore_general_2014}. Note that maximizing $V^\pi_\nu$ for different histories $h_{<t}$ is not in conflict, as can be seen from the Bellman equation. We define $V^*_\pi = V^{\pi^*_\nu}_\nu$ (which does not depend on the choice of $\pi^*_\nu$). Clearly $\pi^*_\nu(\cdot|h_{<t})$ must be supported on $\argmax_a V^*_\nu(h_{<t}a)$, where $V^*_\nu$ is extended naturally to histories ending with an action. Because $\Gamma_t$ is a positive scale factor, maximizing the value function is equivalent to maximizing the expected sum of discounted future rewards. 
\end{definition}
To prove the existence of a \rOc\ Nash equilibrium, we need one more result which is of independent interest.

\begin{theorem}[oracle-computable optimal strategy] \label{thm:O_gen_optimal_strategy}
For any environment $\nu$ whose conditionals are $O$-estimable, and any estimable discount normalization factor $\Gamma_t$, there is a \rOc\ optimal strategy $\pi^*_\nu$.
\end{theorem}
\begin{proof} 
Note that environments do not produce elements of $\cA$ so cannot be completed with $O$; this means that $O$-estimability is stronger than $O$-sampled conditionals (even though they are probability measures). Later we will introduce typed reflective oracles and define reflective-oracle computability for environments; adopting such an oracle closes this gap.

The optimal value function $V_\nu^*$ for $\nu$ with discount factor $\gamma_k$ and discount normalization factor $\Gamma_t = \sum_{i=t}^\infty \gamma_i$, is
\begin{equation}\label{eq:optimal_value}
    V_\nu^*(h_{<t}a_t) ~=~ \frac{1}{\Gamma_t} \lim_{T\to\infty} \sum_{e_t} \max_{a_{t+1}} \sum_{e_{t+1}} ... \max_{a_T} \sum_{e_T} \sum_{i=t}^T \gamma_i r_i \prod_{j=t}^T \nu(e_j | h_{<t} \h_{<j} a_j)  
\end{equation}
We assume that both $\gamma_t$ and $\Gamma_t$ are estimable. This is true in the most common cases that $\gamma_t = \gamma^t$ for a constant rational $\gamma \in (0,1)$ or $\gamma = 1$ until some finite horizon after which it is 0. Our assumption is stronger than estimability of $\gamma_t$ which would only make $\Gamma_t$ l.s.c.\ but estimability of $\Gamma_t$ (for all $t$) immediately implies estimability of $\gamma_t$.

Now all quantities in the limit of equation (\ref{eq:optimal_value}) are estimable. The limit can be approximated from below by iteratively increasing $T$. Recalling that the rewards are bounded to [0,1], this approximation is also within $\Gamma_{T+1}$ of an upper bound because this is the maximum possible return for the rounds after $T$. This means that the limit can be approximated both from above and below hence is estimable. The factor $1/\Gamma_t$ is also estimable when $\Gamma_t > 0$, but when $\Gamma_t = 0$ the unnormalized value function is also 0 and we will not need to estimate it (any action is equally good). 

It would be natural to guess that since the values for each action are $O$-estimable, one can simply compute them to sufficient precision and choose the best. However, this does not deal with ties between action values. Instead we need to take advantage of $O$ access again.  Noting that the value function is in $[0,1]$ we can use \cref{thm:bar_lambda_O_sampled} to construct a TM $T_{\alpha\beta}$ such that
\begin{align*}
& \lambda^O_{T_{\alpha\beta}}(\alpha| \h_{<t}) 
    ~&=~ \fr12 [V^*_\nu(\h_{<t}\alpha) - V^*_\nu(\h_{<t}\beta) + 1] ~\in~ [0;1] \\
    & \lambda^O_{T_{\alpha\beta}}(\beta| \h_{<t}) 
    ~=~ 1 - \lambda^O_{T_{\alpha\beta}}(\alpha| \h_{<t}) ~&=~ \fr12 [V^*_\nu(\h_{<t}\beta) - V^*_\nu(\h_{<t}\alpha) + 1] ~\in~ [0;1]
\end{align*}
where $\alpha$ and $\beta$ are actions.  
Then in the two action case we define
\[
\pi(a |\h_{<t}) ~=~ 
\begin{cases}
    1 \text{ if $a = \alpha$ and $O(T_{\alpha\beta}, \h_{<t}, 1/2) \rightarrow 1 $,} \\
    1 \text{ if $a = \beta$ and $O(T_{\alpha\beta}, \h_{<t}, 1/2) \rightarrow 0 $,} \\
    0 \text{ otherwise.}
\end{cases}
\]
The procedure described above simply calls $O$ once and chooses an action based on the response. Recall the notation $O(T,x,p) \rightarrow 0$ or $O(T,x,p) \rightarrow 1$ indicates that an oracle call with query $(T,x,p)$ yields 0 or 1 (respectively). Since the oracle's behavior is stochastic this does not necessarily mean that $O(T,x,p)$ is valued at 0 or 1.  

When $V^*_\nu(\h_{<t}\alpha) > V^*_\nu(\h_{<t}\beta)$, $\pi$ takes action $\alpha$, and when $V^*_\nu(\h_{<t}\alpha) < V^*_\nu(\h_{<t}\beta)$, $\pi$ takes action $\beta$. When the action values are exactly equal, then $\pi$ randomizes in a fashion depending on $O$, but in this case any action choice is equally good. Because $\pi$ optimizes the optimal value function it is an optimal strategy.

If the action set is larger than 2, as suggested by \cite{fallenstein_reflective_2015}, we can construct a version of $T_{\alpha\beta}$ for each pair of actions, then use $O$ to iteratively compare each action not yet tested against the best so far to find one with the maximum action value.
\end{proof}

\begin{theorem}[Nash equilibrium]\label{thm:refl_oracle_comp_nash}
    For any multi-player game $\sigma$ with l.s.c.\ conditionals, mutually optimal response strategies $\pi_1^*, ..., \pi_n^*$ exist and are reflective-oracle computable.
\end{theorem}

\begin{proof} By \cref{sigma_i_refl}, there is an algorithm to construct $\sigma_i$ from oracle pTM's for $\sigma, \pi_1, ..., \pi_n$. There is also an algorithm to construct $\pi^*_{\sigma_i}$ from $\sigma_i$ following the proof of \cref{thm:O_gen_optimal_strategy}. Combining these two algorithms we obtain an algorithm $T_i$ that constructs $\pi_{\sigma_i}^*$ from $(T_\sigma, T_{\pi_1}, ..., T_{\pi_n})$, following once more the convention that $T_\mu$ is an oracle pTM that samples $\mu$. We now have to show that there are $\pi_i$ such that the constructed optimal responses $\pi_{\sigma_i}^*$ w.r.t.\ environments $\sigma_i^{\pi}$ give back $\pi_i$, i.e.\ that $\exists\pi_i:\pi_{\sigma_i}^*=\pi_i$.
Define $T_i'$ to run the oracle pTM returned by $T_i(T_\sigma, T_1', ..., T_n')$. This relies on the second recursion theorem implicitly: Let pTM $A$ accept a two input pTM $T$ and an input $y$ and construct a new TM $T_y(x) = T(x,y)$. There is a machine $T'(x,i)$ that obtains its own description and runs $N(x)$ where $N = T_i(T_\sigma, A(T',1), ..., A(T',n))$. Formally $T_i' = A(T', i)$. Every step in the process of running $T_i'$ has already been shown to halt, so it samples from the optimal strategy $\pi^*_{\sigma_i}$ (meaning that $N = T_{\pi^*_{\sigma_i}}$). \end{proof}   
Each strategy is optimal given the knowledge of all other players' strategies. Players even act optimally on the histories that they play with probability zero, so this is a subgame perfect Nash equilibrium.
%
\section{Convergence for Bayesian Players}\label{sec:convergence_for_Bayesian}
%
We have shown the existence of a reflective-oracle computable Nash equilibrium, which concerns the case that all players know each other's strategies. It is more interesting to consider Bayesian players that do not know each other's strategies, but only have some belief distribution over possible strategies they may face. It is typically difficult (or impossible) to show convergence for Bayesian players in general environment classes or games; see for example \cite{leike2015bad}. The main obstacle is that players may believe exploration is too dangerous. Kalai and Lehrer \cite{kalai_rational_1993} showed that in an infinitely repeated game with perfect monitoring Bayesian players can learn to play an approximate Nash equilibrium, supporting the centrality of Nash equilibria to game theory\footnote{Or depending on one's perspective, justifying the Bayesian approach to game theory.}. This is a particularly impressive result because it shows convergence for purely rational players (without requiring artificial exploration as in e.g.~Thompson sampling) to a randomized strategy, despite the fact that Bayes optimal strategies can always be made deterministic\footnote{See \cite{foster_impossibility_2001} for an explanation of further difficulties}. Any solution to the grain of truth problem gets around this apparent contradiction because the deterministic Bayes optimal strategies may not appear in $\cP$. The catch is that, informally, each player must assign a small positive probability (a ``grain of truth") to the strategies actually chosen by his opponents. Well known impossibility theorems (\cite{nachbar_prediction_1997,foster_impossibility_2001}) have suggested that this condition is hard or impossible to meet and limited the applicability of Kalai and Lehrer's results. 
Indeed, it took 22 years for the first non-trivial such class to be found \cite{aixi_reflective_2015,fallenstein_reflective_2015,leike_formal_2016}. We will show how reflective oracles can be used to construct a grain of truth by taking advantage of the effective enumeration of $\cP\rOi$ to find a dominant ``mixture'' strategy $\zeta$. It is then straightforward to construct Bayesian players whose beliefs are consistent with any strategy in the rich class $\cP\rOi$ that satisfy the conditions of Kalai and Lehrer's result. Our novel\footnote{Though \cite{leike_formal_2016} suggests that Kalai and Lehrer's conditions can be satisfied with reflective oracles, they do not provide a proof or even explicitly construct the appropriate strategy class $\cP\rOi$. Also, they claim that all players must know the others are Bayesian, which is not required.} result shows that Nash equilibria arise very naturally in infinitely repeated stage games, at least insofar as it is natural to supply players with a common reflective oracle.

\paragraph{Infinitely repeated games of Kalai and Lehrer.}
Kalai and Lehrer require that each player $i$ maintains independent belief distributions over the strategy of all players. Player $i$'s uncertainty about which strategy in $\cP$ player $j$ has chosen can be expressed as mixture of behavior strategies in $\cP$, and is itself a behavior strategy by Kuhn's theorem \cite{aumann_mixed_1964} (though in general it may not be in $\cP$). Therefore we can write it as $\pi^i_j$, with superscript representing the player who's state of knowledge we are considering and the subscript representing the player he is reasoning about, so that player $i$'s full beliefs about the strategies of all players is given by a vector $\pi^i = (\pi^i_1, \pi^i_2, ..., \pi^i_n)$. This also allows us to represent more general beliefs that might not be constructed as a Bayesian mixture over a strategy class. Every player at least knows his own strategy so $\pi^i_i = \pi_i$.
The true strategy vector is given by $\pi_{1:n} = (\pi_1, \pi_2, ..., \pi_n)$ as in \cref{sec:refl_nash_equilibrium} (we will sometimes suppress the subscripts $1:n$ in $\pi_{1:n}$). The lack of a superscript indicates that this is not subjective. We assume that the reward for player $i$ is specified by a fixed payoff function $u_i$ depending only on the actions of all players in the current round. Each player knows his own payoff function (since the action sets are finite, $u_i$ is sometimes called a payoff matrix, and for our purposes may be assumed computable without any significant loss in generality). Though player $i$ does not know any other player's payoff function, that information would not be useful anyway because he does not assume other players' policies to be optimal. Perfect monitoring means that each player observes the other players' actions; there are no further observations. Therefore $\sigma$ is a multi-player environment as defined above but with additional restrictions; in particular there is no longer a meaningful difference between $\sigma^\pi$ and $\sigma^\pi_i$, because $e^i_t = a^{\neq i}_t = a^1_t...a^{i-1}_ta^{i+1}_t...a^n_t$. 
This means that the history distribution $\sigma^\pi$ contains multiple copies of the same action history as distributed to each player through $\sigma^\pi_i$.
It is now the case that a player's beliefs about his subjective environment depend on both his index in $\sigma$ and his beliefs about the strategies of other players, so that player $i$ models his environment as $\smash{\sigma^{\pi^i}_i}$ which is not in general the same as his subjective environment $\sigma^{\pi}_i$. 

\paragraph{Conditions for convergence.}
Kalai and Lehrer's result requires that $\pi_i$ acts rationally with respect to player $i$'s beliefs, or in our terminology that $\pi_i = \pi_{\smash{\sigma^{\pi^i}_i}}^*$. This is not a circular definition because $\smash{\sigma^{\pi^i}_i}$ does not depend on $\pi^i_i$ (when viewed as an environment), see Appendix~\ref{app:subjective_environment}. 
Additionally, they require that $\sigma^\pi \ll \sigma^{\pi^i}$, which follows from the grain of truth property.

\paragraph{Constructing a mixture policy.}
We want to satisfy the convergence conditions of Kalai and Lehrer, but this could be done without learning by setting each player's beliefs $\pi^i$ to the true optimal strategies $\pi^* = (\pi^*_{\sigma_1}, ..., \pi^*_{\sigma_n})$ as in \cref{sec:refl_nash_equilibrium}; for our purposes the players must also hold (independent) priors distributed over all of $\cP^O_\text{refl}$ to model their ignorance of each opponent's strategy. We will actually satisfy a slightly different condition by finding a \emph{dominant} strategy $\zeta \in \cP\rOi$ such that $\forall \pi \in \cP\rOi,$ $\exists c \in \mathbb{R}^+$ such that $\zeta(\cdot) \geq c\pi(\cdot)$. Our usage of the term ``dominant strategy" is not related to the usual game-theoretic meaning; it is a measure-theoretic property not an optimality property. A Bayesian mixture $\sum_{\pi\in\cP\rOi} w_\pi \pi(\cdot)$ satisfies the latter with $c=w_\pi>0$. Bayesian mixtures over $\cP$ are not always in $\cP$ but we show below that this holds if $\sum_\pi w_\pi = 1$, so that we could simply define $\zeta$ this way and it would be a probability measure and therefore a strategy (so the following algorithm is unnecessary). However, the simplicity based priors often used in algorithmic information theory \cite{livitanyi}, including to define Solomonoff induction and AIXI \cite{Hutter:24uaibook2} are only l.s.c.\ semimeasures. The following construction for $\zeta$ encompasses the general case that the weights may be only l.s.c., which only happens when they are defective ($\sum w_\pi < 1$) because l.s.c.\ probability measures are estimable. A player with prior $\zeta$ still learns any opponent's strategy in $\cP^O_\text{refl}$ in the sense of strong merging \cite{kalai_rational_1993}, so $\zeta$ models an unknown strategy, and it can also be used to satisfy Kalai and Lehrer's conditions when all players are Bayesian.
Fix l.s.c.\ weights $w_\pi>0$ for each $\pi \in \cP\rOi$. For any pTM $T$ let $\pi_T$ be the strategy corresponding to the measure $\bar\lambda_T^O$, and consider TM $Q$ implementing Algorithm \ref{alg:Q}.

\paragraph{Algorithm idea.}
We would like to sample from $\zeta' = \sum_\pi w_\pi \pi$, but because we want loose requirements on the computability of $w_\pi$ we cannot assume they sum to 1. This means we would like to complete $\zeta'$. Unfortunately we cannot do this either because though $\zeta'$ is $O$-l.s.c., its conditionals $\zeta'(a_{\leq t} || e_{\leq t})/\zeta'(a_{<t}||e_{<t})$ involve division by an  $O$-l.s.c.\ quantity so may not be $O$-l.s.c.\ themselves; in particular we do not even know if $\zeta'$ is $O$-sampled\footnote{When $w_\pi$ is estimable, the conditionals are estimable and therefore $O$-sampled; in \cite{aixi_reflective_2015} this is elegantly demonstrated by rejection sampling, though it is obvious in light of our \cref{thm:bar_lambda_O_sampled}. Their proof does not extend to l.s.c.\ $w_\pi$.}
\ so we cannot produce an oracle pTM $T$ sampling it and there is no $\lambda^O_T$ to complete. Instead, we define a new pTM $Q$ to lower semicompute the numerator of the conditional, $\zeta'(a_{\leq t} || e_{\leq t})$, and use the completed $\bar{\lambda}^O_Q$ to estimate the denominator $\pi_Q(a_{<t}||e_{<t})$. This makes the fraction $\zeta'(a_{\leq t} || e_{\leq t})/\pi_Q(a_{<t}||e_{<t})$ $O$-l.s.c., which means we can sample from it. Then completing $\lambda^O_Q$ we obtain our measure $\pi_Q$ which may not literally complete $\zeta'$, but does dominate it. The core of this construction relies (again) on Kleene's second recursion theorem, in this case allowing $Q$ access to its own description, which it needs to estimate the denominator, intuitively ``pretending that it has already been completed."\footnote{This iterative completion is reminiscent of Solomonoff normalization \cite[Sec.2.8.2]{Hutter:24uaibook2}, but may not preserve the ratios of conditionals.}
Note that this is the only recursion within $Q$; we estimate $\bar{\lambda}^O_Q$ by running the binary search pTM $B_Q$, which makes oracle calls about $Q$ but never actually simulates $Q$. 

\begin{algorithm} 
	\caption{pTM $Q$}
    \label{alg:Q}
	\begin{algorithmic}[1]
        \Input History $\h_{<t}$ 
        \Require Random sequence $\omega$
        \Output $a_t\sim \lambda_Q^O(a_t | \h_{<t})$
        \State Obtain $\langle Q \rangle$ \label{alg:Q_line1}
        \State Let $\phi_\alpha(\h_{<t}, \cdot)$ approximate $\sum_{\pi \in \cP\rOi} w_\pi \frac{\pi(a_{<t} || e_{<t})}{\pi_Q(a_{<t}||e_{<t})}\pi(\alpha|\h_{<t})$ from below, where $\pi_Q\equiv\bar\lambda_Q^O$ \label{alg:Q_line2}
        \State Run sample($\phi_\alpha$, $\text{\ae}_{<t}$) with access to $\omega$ (\Cref{alg:sample}). 
	\end{algorithmic} 
\end{algorithm}

\paragraph{Algorithm correctness.}
\cref{alg:Q_line1} is possible by Kleene's second recursion theorem. \Cref{alg:Q_line1} is doing most of the work; we need to show that the right hand side is $O$-l.s.c.\ Because $w_\pi$ is assumed l.s.c.\ and $\pi$ and $\pi_Q$ are $O$-estimable by \cref{thm:bar_lambda_est_properties}, every term of the sum is $O$-l.s.c.\ This means that the sum is $O$-l.s.c.\ (computing the $k^\text{th}$ partial sum for $\phi_\alpha(\cdot,k)$). We have to show inductively that the denominator is never 0, but this is easy because there is a computable measure assigning any finite string nonzero probability. 
Therefore $\zeta\in\cP\rOi$. 

By the correctness of the sampling algorithm,
\begin{equation*}
    \lambda_Q^O(a_t | \h_{<t}) ~= \sum_{\pi \in \cP\rOi} w_\pi \frac{\pi(a_{<t} || e_{<t})}{\pi_Q(a_{<t}||e_{<t})}\pi(a_t|\h_{<t})
\end{equation*}
Now we can choose our dominant policy as $\pi_Q$:

\begin{equation*} 
    \zeta ~:=~ \pi_Q ~=~ \bar\lambda_Q^O ~\geq~ \lambda_Q^O
\end{equation*}

By definition $\zeta \in \cP\rOi$. It only remains to show that $\zeta$ dominates the class. 
\begin{align*}
    \zeta(a_{\leq t} || e_{\leq t}) 
    ~&=~ \zeta(a_{<t}||e_{<t}) \zeta(a_t | \h_{<t}) 
    ~\geq~ \zeta(a_{<t}||e_{<t}) \lambda_Q^O(a_t | \h_{<t}) \\
    ~&=~ \zeta(a_{<t}||e_{<t}) \sum_{\pi \in \cP\rOi} w_\pi \frac{\pi(a_{<t} || e_{<t})}{\pi_Q(a_{<t}||e_{<t})}\pi(a_t|\h_{<t}) \\
    ~&= \sum_{\pi \in \cP\rOi} w_\pi \pi(a_{\leq t} || e_{\leq t})
    ~\geq~ w_\pi \pi(a_{\leq t} || e_{\leq t}) ~~~\forall \pi \in \cP\rOi
\end{align*}
noting for the inequality that the $O$-completed measure for any machine is lower bounded by its semi-measure. 
As desired, $\zeta \in \cP\rOi$, and $\zeta \geq \cP\rOi$.
This observation is sometimes written as
%
\begin{theorem}[$\cP\rOi$ contains a dominant element] \label{thm:dominant_zeta}
    There exists $\zeta \in \cP^O_\text{refl}$ s.t. $\forall \pi \in \cP^O_\text{refl}$,\ $\exists c>0$ such that $\zeta(\cdot) \geq c\pi(\cdot)$ ($\zeta$ multiplicatively dominates $\pi$). The first condition is that $\zeta$ is in the class and the second that $\zeta$ dominates the class ($\zeta \geq \cP^O_\text{refl}$). Because both are satisfied we say that $P^O_\text{refl}$ has a dominant element, written $\cP\rOi \geq \cP\rOi$.
\end{theorem}
When the weights $w$ sum to 1, $\zeta\in\cP\rOi$ is a Bayesian mixture over $\cP\rOi$, but we don't actually need this.
We only need the dominance property, and $\zeta$ dominates the defective Bayesian mixture $\sum_{\pi \in \cP^O_\text{refl}} w_\pi \pi$ and consequently any $\pi \in \cP^O_\text{refl}$.  
%
\paragraph{A grain of truth.}
Choosing $\pi^i_j = \zeta$ for $i \neq j$, $\sigma_i^{\pi^i}$ is defined as a product over $O$-sampled policies $\zeta$ (with deterministic computable rewards)\footnote{It does not depend on $\pi_i^i$, making its definition a slight abuse of notation.}. By \cref{sigma_i_refl}, $\sigma_i^{\pi^i}$ has $O$-estimable conditionals, and by \cref{thm:O_gen_optimal_strategy} there is a reflective-oracle computable optimal strategy $\pi^*_{\smash{\sigma^{\pi^i}_i}} \in \cP\rOi$. Therefore, the policy class $\cP\rOi$ and any multi-player game with l.s.c.\ conditionals form a solution to the grain of truth problem. The setting of Kalai and Lehrer (infinitely repeated games with computable rewards) is a special case: Because $\forall j ~\pi^*_{\smash{\sigma^{\pi^j}_j}} \ll \zeta$, it is easy to see that when $\pi = (\pi^*_{\smash{\sigma^{\pi^1}_1}}, \pi^*_{\smash{\sigma^{\pi^2}_2}}, ..., \pi^*_{\smash{\sigma^{\pi^n}_n}})$, and $\pi^i = (\zeta, ..., \pi^*_{\smash{\sigma^{\pi^i}_i}}, ..., \zeta)$, $\sigma^\pi \ll \smash{\sigma^{\pi^i}}$. Finally, we can appeal to \cite[Thm. 2]{kalai_rational_1993} to conclude the following: 
%
\begin{theorem}[close to $\eps$-Nash equilibrium]\label{thm:kalai_lehrer_convergence}
    In a computable infinitely repeated game $\sigma$, if $\pi^i_j = \zeta$ for $i \neq j$ and $\pi^i_i = \pi_i = \pi^*_{\smash{\sigma^{\pi^i}_i}}$  (so all players are Bayesian), then for every $\eps > 0$, $\sigma^\pi$-a.s.\ there is a time $t_\eps$ such that for all $t \geq t_\eps$, $\sigma^\pi$ plays $\eps$-like the history distribution of a $\eps$-Nash equilibrium.
\end{theorem}

The term $\eps$-Nash equilibrium means that every players' expected utility is within $\eps$ of the best achievable given knowledge of the other players' strategies. The term ``plays $\eps$-like'' is defined in \cite[Def. 2]{kalai_rational_1993}, relying on \cite[Def. 1]{kalai_rational_1993} of ``$\eps$-close" measures. It means that with high probability the conditionals of the history distributions are close; Kalai and Lehrer point out this is a kind of Provably Approximately Correct (PAC) guarantee. Statements holding after time $t_\eps$ refer to measures conditioned on the history up to time $t_\eps$. 

\paragraph{Matching pennies example.}
In the game of \emph{matching pennies} there are two agents ($n = 2$),
and two actions $\cA = \{ \alpha, \beta \}$
representing the two sides of a penny.
In each time step
agent $1$ wins if the two actions are identical and
agent $2$ wins if the two actions are different.
The payoff matrix is as follows.
\begin{center}
\begin{tabular}{l|cc}
         & $\alpha$ & $\beta$ \\
\hline
$\alpha$ & 1,0      & 0,1 \\
$\beta$  & 0,1      & 1,0
\end{tabular}
\end{center}
We use $\cE = \{ 0, 1 \}$ to be the set of rewards
(observations are vacuous) and define the multi-agent environment $\sigma$
to give reward $1$ to agent $1$ iff $a_t^1 = a_t^2$ ($0$ otherwise) and
reward $1$ to agent $2$ iff $a_t^1 \neq a_t^2$ ($0$ otherwise).

According to our result, when the game is known Bayesian players with prior $\cP^O_\text{refl}$ eventually converge to a Nash equilibrium of the repeated game. When discounting is geometric with discount factor $\gamma$ close to 0, this means they will approximately play the (only) Nash equilibrium of the stage game, randomizing uniformly between actions $\alpha$ and $\beta$.

See \cite{brand_imp_2016} for an early discussion of this game in the context of computability theory. 

%
\paragraph{Differing action sets.}
When each player has a different action set, in order to use a consistent reflective oracle for every player, it is necessary to consider the encodings of actions. Leike et al.'s (implicit) approach \cite{leike_formal_2016} was to choose complete, prefix-free, binary codes for each type of symbol, which introduces no serious difficulties but means that their algorithms should properly be specified on the bit level. We would like to use non-binary reflective oracles to take a more elegant approach. The naive idea of combining the action sets for each player $\cA = \bigsqcup_{1\leq i \leq n} \cA_i$ does not immediately work because player $i$'s strategies must be have their conditionals completed to probability measures in $\Delta\cA_i\neq\Delta\cA$. The solution is slightly harder than using $n$ separate reflective oracles, because the $n$ reflective oracles would have to consistently answer queries about each other. Fortunately it is possible to use a \emph{typed} reflective oracle as described in \cref{sec:non_binary_refl_oracle_existence} to map actions from each action set to their own simplex.

\section{Impossibility Results}
\label{sec:impossibility-results}

Why does our solution to Kalai and Lehrer's grain of truth problem
not violate the impossibility results from the literature?
Assume we are playing an infinitely repeated game where
in the stage game
no agent has a weakly dominant action and
the pure action maxmin reward is strictly less then the minmax reward.
The impossibility result of \cite{nachbar_prediction_1997, nachbar_beliefs_2005}
state that
there is no class of policies $\cP$ such that
the following are simultaneously satisfied.
\begin{itemize}
\item \emph{Learnability.}
	Each agent learns to predict the other agent's actions.
\item \emph{Caution and Symmetry.}
	The set $\cP$ is closed under simple policy modifications
	such as renaming actions.
\item \emph{Purity.}
	There is an $\varepsilon > 0$ such that
	for any stochastic policy $\pi \in \cP$
	there is a deterministic policy $\pi' \in \cP$ such that
	if $\pi'(\h_{<t}) = a$, then $\pi(a | \h_{<t}) > \varepsilon$.
\item \emph{Consistency.}
	Each agent always has an $\varepsilon$-best response available in $\cP$.%
\end{itemize}
In order to converge to an $\varepsilon$-Nash equilibrium,
each agent has to have an $\varepsilon$-best response available to them,
so consistency is our target.
Learnability is immediately satisfied for any environment in our class
if we have a dominant prior \cite{kalai_rational_1993}.
For $\cP^O_\text{refl}$ caution and symmetry are also satisfied
since this set is closed under any computable modifications to policies.
However, our class $\cP^O_\text{refl}$ avoids this impossibility result because
it violates the purity condition:
Let $T_1, T_2, \ldots$ be an enumeration of $\mathcal{T}$.
With action space $\cA = \{0, 1\}$, consider the policy $\pi$
that maps history $\h_{<t}^i$ to the action $1 - \flip(O_1(T_t, \h^i_{<t}, 1/2))$.
If $T_t$ is deterministic,
then $\pi$ will take a different action than $T_t$
for any history of length $t - 1$.
Therefore no deterministic reflective-oracle-computable policy
can take an action that
$\pi$ assigns positive probability to in every time step.

\cite{foster_impossibility_2001} present a condition that
makes convergence to a Nash equilibrium impossible:
if the player's rewards are perturbed by a small real number
drawn from some continuous density $\nu$,
then for $\nu$-almost all realizations the players do not learn to
predict each other and do not converge to a Nash equilibrium.
For example, in a matching pennies game,
rational agents randomize only if the (subjective) values of both actions
are exactly equal.
But this happens only with $\nu$-probability zero, since $\nu$ is a density.
Thus with $\nu$-probability one the agents do not randomize.
If the agents do not randomize,
they either fail to learn to predict each other,
or they are not acting rationally according to their beliefs:
otherwise they would seize the opportunity to
exploit the other player's deterministic action.

But this does not contradict our convergence result:
the class $\cP^O_\text{refl}$ is countable and each $\nu \in \cP^O_\text{refl}$
has positive prior probability.
Perturbation of rewards with arbitrary real numbers is not possible.
Even more, the argument given by \cite{foster_impossibility_2001}
cannot work in our setting:
the Bayesian mixture $\pi_Q$ mixes over $\lambda_T$
for all probabilistic Turing machines $T$.
For Turing machines $T$ that sometimes do not halt,
the oracle decides how to complete
$\lambda_T$ into a measure $\bar{\lambda}^O_T$.
Thus the oracle has enough influence on the exact values in the Bayesian mixture
that the values of two actions in matching pennies can be made exactly equal.

\section{Asymptotic Optimality in Unknown Games} \label{sec:unknown_games}

We now go further and show convergence to equilibrium even when the game is unknown and is not-repeated but one infinitely long game, as long as the players use \emph{asymptotically optimal} strategies instead of Bayes-optimal strategies.

\begin{definition}[asymptotic optimality]\label{def:asymptotic_optimality}
    A policy $\pi$ is asymptotically optimal in mean in environment class $\cM$ iff $\forall\mu\in\cM$,\ $\mathbb{E}^\pi_\mu [V^*_\mu(h_{<t}) - V^\pi_\mu(h_{<t})] \rightarrow 0$ as $t \rightarrow \infty$. 
\end{definition}

In fact, we can show convergence even when the players are not initially aware of each others' existence. To do this, we must extend the environment class to all \rOc\ environments (which we define below similarly to $\cP^O_\text{refl}$). Because we would still like players to be included in the environment we need the oracle to be usable for computing either strategies or multi-player games. Then the entire arrangement of a multi-player game with the other players embedded is \rOc. The situation is similar to differing action sets; it would be possible to simply give environments access to a reflective oracle that provides completed action probabilities, but unfortunately this could not be used to complete the action-conditional semimeasures generated by pTM's producing perceptions, which ultimately means that the environment class would not be effectively enumerable. Instead we use a typed reflective oracle capable of completing both percept and action distributions. After defining our environment class, we show that Thompson sampling policies converge to a Nash equilibrium.

\begin{definition}[$\cM^O_\text{refl}$]\label{def:_M_refl}
    Fix alphabet $\Sigma = \cA \sqcup \cE$. Let $O$ be a typed reflective oracle with input and output alphabet $\Sigma$. Let $\cM\rOi$ be the set of  environments $\nu_T = \bar{\lambda}^O_T$, called \rOc. 
\end{definition}
Unlike the class of environments or games with estimable\ conditionals, $\cM\rOi$ is effectively enumerable because halting issues are resolved by oracle completion. 
%
\begin{theorem}[convergence to equilibrium] \label{thm:asymptotically_opt_convergence}
    Let $\sigma$ be a reflective-oracle computable multi-agent game and let $\pi_1, ..., \pi_n$ be \rOc\ policies that are asymptotically optimal in mean in the class $\cM\rOi$. Then for all $\eps > 0$ and all $i \in \{ 1, ..., n \}$ the $\sigma^{\pi_{1:n}}$-probability that the policy $\pi_i$ is an $\eps$-best response converges to 1 as $t\to\infty$. 
\end{theorem}
\begin{proof}
    This is \cite[Thm. 28]{leike_formal_2016}. Following the argument of \cref{sigma_i_refl}, subjective environments are in $\cM\rOi$. Because each policy is asymptotically optimal in mean in its subjective environment, \cref{thm:asymptotically_opt_convergence} follows from the observation that convergence in mean implies convergence in probability for bounded random variables. Therefore,
    \[
    \sigma^\pi_i[V^*_{\sigma_i}(\text{\ae}^i_{<t}) - V^{\pi_i}_{\sigma_i}(\text{\ae}^i_{<t}) \geq \eps] \rightarrow 0 \text{\ as \ } t \rightarrow \infty
    \]
    so the probability that $\pi_i$ plays an $\eps$-best response converges to 1 as $t \rightarrow \infty$.
\end{proof}

\paragraph{Thompson sampling.}
Now we only need to find a set of asymptotically optimal \rOc\ strategies. 
It is normally assumed that Bayesian agents solve the exploration/exploitation problem in a principled way,
so it is somewhat surprising that they are not (even weakly) asymptotically optimal in too general environment classes \cite{Orseau:10}.
Thompson sampling \cite{thompson1933likelihood,Hutter:16thompgrl} is an asymptotically optimal variation of the Bayesian rational strategy modified to increase exploration. Let the effective horizon $H_t(\eps)$ be the minimum number of steps in the future such that the discount normalization factor is less than a $\eps$ fraction of the current discount normalization factor; in the case $\eps = 1/2$ this is the ``half-life" of $\Gamma_t$. Formally $H_t(\eps) = \min_{k}\{k|\Gamma_{t+k}/\Gamma_t \leq \eps\}$. Then Thompson sampling is described by \cref{alg:pi_TS} and denoted by $\pi_\TS$. Naturally it is parameterized by a class of environments and a p(oste)rior weight function $w$. We will choose $\cM\rOi$ for the class of environments, which means players initially are not even aware that their opponents exist (or of their number). Then the first condition of the strong grain of truth property is satisfied:
\begin{theorem}[Strong grain of truth property, first condition]\label{thm:M_includes_G_P}
    If $\sigma$ is an $O$-estimable game and $\pi_{\neq i} \in (\cP^O_\text{refl})^{n-1}$ then $\sigma^\pi_i \in \cM^O_\text{refl}$.
\end{theorem}
\begin{proof}
    This follows from a slight generalization of Theorem~\ref{sigma_i_refl} to include $O$-estimable games.
\end{proof}
If more prior knowledge is required we can instead enumerate $\cG\times\cP^{n-1}$. Note that even in this case $w$ is a \emph{joint} p(oste)rior over $\cM=\cG\times\cP^{n-1}$, i.e.\ can model any collusion between opponents (and even the game itself).

If we only assume that the weights are lower semicomputable semimeasures, there is a chance that sampling from them fails. This means that the infinite loop may get stuck after finitely many iterations. It seems that Thompson sampling is not \rOc\ (because it is not a contextual probability measure) without stronger assumptions on $w$, for instance $O$-estimability. We can rephrase the Thompson sampling algorithm as shown in \cref{alg:stepwise_pi_TS} to explicitly show how to sample actions on each step (reflective-oracle computably) without persistent memory instead of abstractly describing the behavior between resampling environments. Equivalence is similar to Kuhn's theorem \cite{aumann_mixed_1964}.
\begin{tabular}{cc}

\begin{minipage}[t]{0.36\textwidth}
\begin{algorithm}[H]
\caption{Thompson sampling strategy $\pi_\TS$} 
\label{alg:pi_TS}
\begin{algorithmic}[1]
    \Input Percept stream $e_{1:\infty}$
    \Output $a_{1:\infty} \sim \pi_\TS(\cdot||e_{1:\infty})$
    \While{true}
    \State sample $\rho \sim w( \cdot | \h_{<t})$
    \State follow $\pi_\rho^*$ for $H_t(\eps_t)$ steps
    \EndWhile
\end{algorithmic}
\end{algorithm}
\end{minipage}
&
\begin{minipage}[t]{0.6\textwidth}
\begin{algorithm}[H]
	\caption{Stepwise Thompson sampling strategy $\pi_\TS$} 
    \label{alg:stepwise_pi_TS}
	\begin{algorithmic}[1]
        \Input History $\h_{<t}$
        \Output $\pi_\TS(a_t|\h_{<t})~\forall a_t\in\cA$
        \State $t_0 \gets 0$; $i \gets 0$
        \While{$t_i \leq t$} \{ $t_{i+1} \gets t_i + H_{t_i}(\eps_{t_i})$; $i \gets i+1$ \}
        \EndWhile
        \State $t' \gets t_{i-1}$
        \State $\displaystyle\pi_\TS(a_t|\h_{<t}) := \sum_{\nq\rho\in\cM\rOi\nq} w( \rho | \h_{<t'}) \frac{\pi^*_\rho(\h_{t':t}|\h_{<t'})}{\pi_\TS(\h_{t':t}|\h_{<t'})} \pi_\rho^*(a_t|\h_{<t})$
    \end{algorithmic} 
\end{algorithm}
\end{minipage}
     
\end{tabular}

\hfill

It still remains to show that \cref{alg:stepwise_pi_TS} is \rOc. By definition,
\begin{equation*} 
    w(\rho|\h_{<t}) ~=~ w(\rho) \frac{\rho(\h_{<t})}{\xi(\h_{<t})}
    ~~~\text{where}~~~ \xi(\h_{<t}) ~:=~ \sum_{\rho\in\cM\rOi} w(\rho) \rho(\h_{<t})
\end{equation*}
When $w$ is $O$-estimable, also every factor above (assuming that the environment class is general enough that all finite history prefixes are possible) and therefore the posterior weights are $O$-estimable. 
For any $\rho \in \cM\rOi$, \cref{thm:O_gen_optimal_strategy} shows that the optimal policies $\pi_\rho^*$ are all $O$-estimable. This makes Line 4 of \cref{alg:stepwise_pi_TS} possible with $O$ access.
%
\begin{theorem}[Thompson sampling computability] \label{thm:est_pi_TS_comp}
    For estimable $\Gamma_t$ and normalized estimable prior $w$ over over $\cM\rOi$, $\pi_\TS$ over $\cM\rOi$ is \rOc.
\end{theorem}
Together, \cref{thm:M_includes_G_P} and \cref{thm:est_pi_TS_comp} show that $\cP^O_\text{refl}$ and the class of $O$-estimable games satisfy the Thompson sampling version of the strong grain of truth property.

Technically, computing the horizon $H_t(\epsilon_t)$ exactly would require finitely computable $\Gamma_t$ and $\epsilon_t$, but the convergence of Thompson sampling only depends on $\epsilon_t > 0$ and $\epsilon_t \rightarrow 0$. As long as $\Gamma_{t+k}/\Gamma_t$ is computed to sufficient precision to ensure the ratio between resampling steps decreases to 0 this is equivalent to Thompson sampling with an acceptable choice of $\epsilon_t$.  

\paragraph{(Un)normalized weights $w$.}
Recall that all normalized l.s.c.\ $w$ are also estimable. Given that we generally assume in this paper that the weights are at least l.s.c., \cref{thm:est_pi_TS_comp} only relies on the weights summing to one; without this requirement, Thompson sampling would sometimes fail to sample an environment and its behavior is under-specified!  
Generalizing to l.s.c.\ weights, it is natural to try to use the reflective oracle to somehow complete Thompson sampling. We could try to complete $\pi_\TS$'s environment mixture $\xi$. Unfortunately this would not explicitly complete the weights which Thompson sampling needs access to; $\pi_\TS$ requires not a dominant environment but explicit coefficients. The reflective oracle could be used to directly complete each weight from an oracle pTM generating it (in the sense of outputting 1 with probability $w(\rho | \h_{<t})$ and otherwise failing to halt) but it is unclear whether the individually completed weights would still sum to 1.\footnote{Really, what we would like to do is divide $\xi$ by $\sum_\rho w(\rho)$ to normalize directly, but this is not even l.s.c.}
\newline

Combined with \cref{thm:est_pi_TS_comp}, \cref{thm:O_limit_comp} tells us that we can choose $O$ to make Thompson sampling limit-computable, which lets us improve \cref{thm:asymptotically_opt_convergence}.
%
\begin{theorem}[limit-computable convergence to equilibrium {\cite[Cor.20]{leike_formal_2016}}] \label{thm:limit_comp_nash_convergence}
    There are limit-computable strategies $\pi_1, ..., \pi_n$ such that for any computable multi-agent game $\sigma$ and for all $\eps > 0$ and all $i \in \{ 1, ..., n \}$ the $\sigma^{\pi_{1:n}}$-probability that the policy $\pi_i$ is an $\eps$-best response converges to 1 as $t\to\infty$.
\end{theorem}
Since all $\pi_i$ converge to $\eps$-best responses, this implies that $\pi_{1:n}$ is asymptotically $\eps$-Nash. We can say more about the computability level of $\cM\rOi$:
%
\begin{theorem}[$\Delta_1 \subset \cM\rOi \subset \Delta_2$]\label{thm:M_refl_computability}
    The class $\mathcal{M}\rOi$ contains all (joint) estimable (normalized) environments (sometimes called $\cM^\text{msr}_\text{est}$) and is contained in the class of measures with limit-computable conditionals.
\end{theorem}

\begin{proof}
    The claim $\cM^\text{msr}_\text{est} \subset \cM\rOi$ follows immediately from \cref{lemma:est_O_sampled}, and is in fact strict by a simple diagonalization argument in \cite{leike_formal_2016}. The claim that $\cM\rOi \subset \Delta_2$ follows from \cref{thm:O_limit_comp}. It is easy to see that $\Delta_1 \subset \mathcal{P}\rOi \subset \Delta_2$ also holds by the same argument.
\end{proof}
%
\section{An Application to Self-Prediction} \label{sec:self_prediction}

In \cref{sec:unknown_games} we derived a convergence result for players who are not initially aware of the multi-player game they are playing, or even that other players are involved. We can take this ignorance even further by allowing our player to be uncertain of even his own strategy as he selects each individual move. This is the setting of the Self-AIXI agent proposed by \cite{catt_self-predictive_2023}. Though the problem may appear esoteric at first, it is of interest to reinforcement learning (RL) researchers. Model-based RL algorithms often execute an expensive decision tree search to plan their future actions. This search can be narrowed by (iteratively) distilling the resulting policy into a model that guides action selection. The Self-AIXI agent can be viewed as an extreme case of this approach, entirely replacing planning with an interaction between self-model and environment-model akin to model- and value-based policy search methods but more general and principled. Arguably, the Self-AIXI framework also describes human planning; though we may make plans for our future actions, we do not know which strategy we will ultimately follow. 
Formally, a Self-AIXI policy is defined\footnote{We require strategies and environments to be (contextual, chronological) probability measures. This definition generalizes from explicit Bayesian mixtures to any dominant elements of each class. Also, our \cref{def:self_AIXI} does not maximize equation (3) of \cite{catt_self-predictive_2023}, which asserts linearity of the action value function $\smash{Q^\zeta_\xi}$ with the incorrect coefficients (failing to update on the latest action) and is probably not the intended definition as it is inconsistent with the rest of the paper. This definition of $\smash{Q^\zeta_\xi}$ would make Self-AIXI a kind of one-step causal decision theorist instead of an evidential decision theorist.} as
%
\begin{definition}[Self-AIXI]\label{def:self_AIXI}
Let $\zeta$ and $\xi$ be dominant elements of policy class $\cP$ and environment class $\cM$. 
  \begin{align*}
    \pi_S(h_{<t}) ~&\in~ \argmax_{a_t\in\cA} V_\xi^\zeta(h_{<t}a_t) \\
    V_\xi^\zeta(h_{<t}a_t) ~&:=~ \frac{1}{\Gamma_t} \lim_{m\to\infty} \sum_{a_{t+1:m}, e_{t:m}} \sum_{i=t}^m \gamma_i r_i \prod_{j=t}^m \xi(e_j|\h_{<j}a_j)\prod_{j=t+1}^m\zeta(a_j|\h_{<j})       
  \end{align*}
\end{definition}
The results of \cite{catt_self-predictive_2023} suggest convergence of Self-AIXI to $\smash{\pi^*_\xi}$ (when $\xi$ is a dominant element of the class $\cM^\text{semi}_\text{lsc}$ of environments given by contextual chronological l.s.c.\ semimeasures, $\pi^*_\xi$ is called AIXI), but there are gaps remaining\footnote{We do not address their requirement that $\pi_S$ is ``reasonable off-policy'' which is a technical and slightly unnatural condition that has not been shown for any combination of strategy and environment classes.}. One problem is that they rely on $\pi_S \in \cP$ without constructing any policy class (interesting or otherwise) with this property. Reflective oracles provide a natural example. 
%
\begin{theorem}[Self-AIXI in computable environment] \label{thm:self_aixi_comp_env}
    Let $O$ be a reflective oracle with input alphabet $\Sigma = \cA \sqcup \cE$ and output alphabet $\cA$. Let $\cP = \cP\rOi$ and $\cM$ be any environment class containing a dominant element $\xi$ with estimable\ conditionals. Then there exists a dominant $\zeta \in \cP\rOi$, and there is a stochastic $\pi_S \in \cP\rOi$ for $\zeta$, $\xi$.
\end{theorem}
\begin{proof} The existence of such a $\zeta$ follows from \cref{thm:dominant_zeta}. The argument that $\pi_S \in \cP\rOi$ follows \cref{thm:O_gen_optimal_strategy}; the conditionals of $\xi$ are assumed estimable and the conditionals of $\zeta$ are $O$-estimable.
\end{proof} 
Note that since we do not have an effective enumeration of the class of environments with estimable conditionals, we cannot choose it as $\cM$ in \cref{thm:self_aixi_comp_env}, making the result somewhat less interesting. We can solve this problem by sharing the reflective oracle between the strategies and the environments as in \cref{sec:unknown_games}:
%
\begin{theorem}[Self-AIXI in oracle computable environment]\label{thm:self_aixi_oracle_computable_env}
    Fix some finite alphabet $\Sigma$ and encodings over $\Sigma$ for each element of $\cA$ and $\cE$. Let $O$ be a typed reflective oracle with input and output alphabet $\Sigma$. Let $\zeta$ and $\xi$ be dominant elements of $\cP\rOi$ and $\cM\rOi$. respectively. Then there is a stochastic $\pi_S \in \cP\rOi$ for $\zeta, \xi$.
\end{theorem}
\begin{proof} In this case the conditionals of $\xi$ are $O$-estimable because it is in $\cM\rOi$, and as before the argument follows \cref{thm:O_gen_optimal_strategy}.
\end{proof}
%
\section{Conclusion}

We have constructed a policy class meeting the requirements of \cite{kalai_rational_1993}. This can be seen as a justifying the importance of Nash equilibria from a Bayesian perspective. When a Bayesian player's strategy is private information, we see no convincing reason for him to play a strategy corresponding to any Nash equilibrium. For instance, a Bayesian who has observed that his opponent in a game of ``rock, paper, scissors" usually chooses ``rock" (in past otherwise similar games) might naturally choose ``paper," instead of randomizing uniformly as anticipated by classical game theory. However, when the game and policy classes satisfy the grain of truth property, players eventually converge to a set of strategies close to a $\eps$-Nash equilibrium.

\paragraph{Reflective oracles in the real world.}
It is interesting to consider whether the grain of truth property is a reasonable assumption about the beliefs and strategies of humans. The limit-computability of (some) reflective oracles makes this at least distantly plausible. Certainly humans should model others as demonstrating roughly comparable computational power to ourselves, and our computational boundedness means any mutual recursion must eventually terminate. Reflective oracles are one possible model of this termination, but it is hard to see how the (critical) assumption that all players use the same reflective oracle can be justified. At least in the self-predictive case it is sensible for an agent to use the same reflective oracle to model their beliefs about themselves and their environment; convergence results should still hold when the true environment does not actually use oracle access, as long as it is in $\cM\rOi$. More speculatively, perhaps a shared reflective oracle is a reasonable assumption for (semi-)cooperative multi-agent systems such as members of a common culture or subagents within a cognitive architecture.  

\paragraph{Future work.}
It would be interesting to determine the degree of centrality and uniqueness of reflective oracles (and the corresponding $\cP\rOi$) among the solutions to the grain of truth problem, and to study in general the connection between the grain of truth problem and self-prediction. Another fascinating research direction is to clarify the computability properties of reflective oracles; we know that there exist limit-computable reflective oracles, but not all reflective oracles are limit-computable because any measure is computable with respect to some reflective oracle (see \cref{sec:non_binary_refl_oracle_existence}). Therefore, it is easy to understand $\bigcup_O \cP\rOi$. We showed in \cref{lemma:est_O_sampled} that regardless of the chosen reflective oracle $O$, all estimable measures $\nu$\footnote{Here we mean that $\nu(x)$ should be estimable, not necessarily the conditional $\nu(\alpha|x)$.} are sampled by some probabilistic Turing machine with access to $O$ (by binary search), but are any other measures in the intersection $\bigcap_O \cP\rOi$?   

\section{Acknowledgements}
This work was supported in part by a grant from the Long-Term Future Fund (EA Funds - Cole Wyeth - 9/26/2023). We would also like to thank Clemens Possnig and Benya Fallenstein for helpful feedback.


\begin{thebibliography}{HQC24}

\bibitem[Ale23]{alexander_private_2023}
Samuel~A. Alexander.
\newblock Private {Memory} {Confers} {No} {Advantage}.
\newblock In {\em Software {Engineering} and {Formal} {Methods}. {SEFM} 2023
  {Collocated} {Workshops}: {CIFMA} 2023 and {OpenCERT} 2023, {Eindhoven},
  {The} {Netherlands}, {November} 6–10, 2023, {Revised} {Selected} {Papers}},
  pages 42--53, Berlin, Heidelberg, November 2023. Springer-Verlag.

\bibitem[Aum64]{aumann_mixed_1964}
Robert J~. Aumann.
\newblock {\em 28. Mixed and Behavior Strategies in Infinite Extensive Games},
  pages 627--650.
\newblock Princeton University Press, Princeton, 1964.

\bibitem[BD16]{brand_imp_2016}
Michael Brand and David~L. Dowe.
\newblock The {IMP} game: {Learnability}, approximability and adversarial learning beyond $\sigma^0_1$, February 2016.
\newblock arXiv:1602.02743 [cs].

\bibitem[CGMH{\etalchar{+}}23]{catt_self-predictive_2023}
Elliot Catt, Jordi Grau-Moya, Marcus Hutter, Matthew Aitchison, Tim Genewein, Grégoire Delétang, Kevin Li, and Joel Veness.
\newblock Self-{Predictive} {Universal} {AI}.
\newblock {\em Advances in Neural Information Processing Systems}, 36:27181--27198, December 2023.

\bibitem[Fan52]{fan1952fixed}
Ky~Fan.
\newblock Fixed-point and minimax theorems in locally convex topological linear spaces.
\newblock {\em Proceedings of the National Academy of Sciences}, 38(2):121--126, 1952.

\bibitem[FST15]{aixi_reflective_2015}
Benja Fallenstein, Nate Soares, and Jessica Taylor.
\newblock Reflective {Variants} of {Solomonoff} {Induction} and {AIXI}.
\newblock In Jordi Bieger, Ben Goertzel, and Alexey Potapov, editors, {\em Artificial {General} {Intelligence}}, pages 60--69, Cham, 2015. Springer International Publishing.

\bibitem[FTC15]{fallenstein_reflective_2015}
Benja Fallenstein, Jessica Taylor, and Paul~F. Christiano.
\newblock Reflective {Oracles}: {A} {Foundation} for {Classical} {Game} {Theory}, August 2015.
\newblock arXiv:1508.04145 [cs].

\bibitem[FY01]{foster_impossibility_2001}
Dean~P. Foster and H.~Peyton Young.
\newblock On the impossibility of predicting the behavior of rational agents.
\newblock {\em Proceedings of the National Academy of Sciences}, 98(22):12848--12853, October 2001.

\bibitem[HQC24]{Hutter:24uaibook2}
Marcus Hutter, David Quarel, and Elliot Catt.
\newblock {\em An Introduction to Universal Artificial Intelligence}.
\newblock Chapman \& Hall/CRC Artificial Intelligence and Robotics Series. Taylor and Francis, 2024.

\bibitem[Hut05]{Hutter:04uaibook}
Marcus Hutter.
\newblock {\em Universal Artificial Intelligence: Sequential Decisions based on Algorithmic Probability}.
\newblock Springer, Berlin, 2005.
\newblock 300 pages, http://www.hutter1.net/ai/uaibook.htm.

\bibitem[KL93a]{kalai_rational_1993}
Ehud Kalai and Ehud Lehrer.
\newblock Rational {Learning} {Leads} to {Nash} {Equilibrium}.
\newblock {\em Econometrica}, 61(5):1019--1045, 1993.

\bibitem[KL93b]{kalai_subjective_1993}
Ehud Kalai and Ehud Lehrer.
\newblock Subjective {Equilibrium} in {Repeated} {Games}.
\newblock {\em Econometrica}, 61(5):1231--1240, 1993.

\bibitem[Kle52]{Kleene1952-KLEITM}
Stephen~Cole Kleene.
\newblock {\em Introduction to Metamathematics}.
\newblock P. Noordhoff N.V., Groningen, 1952.

\bibitem[LH14]{lattimore_general_2014}
Tor Lattimore and Marcus Hutter.
\newblock General time consistent discounting.
\newblock {\em Theoretical Computer Science}, 519:140--154, January 2014.

\bibitem[LH15]{leike2015bad}
Jan Leike and Marcus Hutter.
\newblock Bad {Universal} {Priors} and {Notions} of {Optimality}.
\newblock In {\em Proceedings of {The} 28th {Conference} on {Learning} {Theory}}, pages 1244--1259. PMLR, June 2015.

\bibitem[LLOH16]{Hutter:16thompgrl}
Jan Leike, Tor Lattimore, Laurent Orseau, and Marcus Hutter.
\newblock Thompson sampling is asymptotically optimal in general environments.
\newblock In {\em Proc. 32nd International Conf. on Uncertainty in Artificial Intelligence ({UAI'16})}, pages 417--426, New Jersey, USA, 2016. AUAI Press.

\bibitem[LTF16]{leike_formal_2016}
Jan Leike, Jessica Taylor, and Benya Fallenstein.
\newblock A formal solution to the grain of truth problem.
\newblock In {\em Proceedings of the Thirty-Second Conference on Uncertainty in Artificial Intelligence}, {UAI}'16, pages 427--436. {AUAI} Press, 2016.

\bibitem[LV{\etalchar{+}}08]{livitanyi}
Ming Li, Paul Vit{\'a}nyi, et~al.
\newblock {\em An introduction to Kolmogorov complexity and its applications}, volume~3.
\newblock Springer, 2008.

\bibitem[Nac97]{nachbar_prediction_1997}
John~H. Nachbar.
\newblock Prediction, {Optimization}, and {Learning} in {Repeated} {Games}.
\newblock {\em Econometrica}, 65(2):275--309, 1997.

\bibitem[Nac05]{nachbar_beliefs_2005}
John~H. Nachbar.
\newblock Beliefs in {Repeated} {Games}.
\newblock {\em Econometrica}, 73(2):459--480, 2005.

\bibitem[Ors10]{Orseau:10}
L.~Orseau.
\newblock Optimality issues of universal greedy agents with static priors.
\newblock In {\em Proc. 21st International Conf. on Algorithmic Learning Theory ({ALT'10})}, volume 6331 of {\em LNAI}, pages 345--359, Canberra, Australia, 2010. Springer.

\bibitem[Rud91]{rudin1991functional}
W.~Rudin.
\newblock {\em Functional Analysis}.
\newblock International series in pure and applied mathematics. McGraw-Hill, 1991.

\bibitem[Tho33]{thompson1933likelihood}
William~R Thompson.
\newblock On the likelihood that one unknown probability exceeds another in view of the evidence of two samples.
\newblock {\em Biometrika}, 25(3-4):285--294, 1933.

\bibitem[vNMR44]{von_neumann_theory_1944}
John von Neumann, Oskar Morgenstern, and Ariel Rubinstein.
\newblock {\em Theory of {Games} and {Economic} {Behavior} (60th {Anniversary} {Commemorative} {Edition})}.
\newblock Princeton University Press, 1944.

\bibitem[WH25]{wyeth_hutter_semimeasures_2025}
Cole Wyeth and Marcus Hutter.
\newblock Value under ignorance in universal artificial intelligence.
\newblock {\em Artificial General Intelligence}, 2025.

\end{thebibliography}

\newcommand{\etalchar}[1]{$^{#1}$}

\appendix\section{Appendix}\label{app}

\section{List of Notation}\label{app:Notation}
Our notation follows \cite{leike_formal_2016} as closely as possible.
\begin{tabbing}
  \hspace{0.13\textwidth} \= \hspace{0.73\textwidth} \= \kill
  $[\![R]\!]$            \> = 1 if $R$=true and =0 if $R$=false (Iverson bracket)                \\[0.5ex]
  $:=$ \> defined to be equal \\[0.5ex]
  $\mathbb{N}$ \> the natural numbers \\[0.5ex]
  $\mathbb{Q}$ \> the rational numbers \\[0.5ex]
  $\mathbb{R}$ \> the real numbers \\[0.5ex]
  $t$ \> (current) interaction (time) step, $t\in\mathbb{N}$ \\[0.5ex]
  $i$ \> player=agent index $\in\{1,...,n\}$ \\[0.5ex]
  $j$ \> the other players=agents $\neq i$ part of the environment \\[0.5ex]
  $k$ \> natural number indices \\[0.5ex]
  $n$ \> number of agents=players \\[0.5ex]
  $p,q$ \> a real value with interpretation as a probability. \\[0.5ex]
  $\mathcal{X}^*$ \> the set of all finite strings over the alphabet $\mathcal{X}$ \\[0.5ex]
  $\mathcal{X}^\infty$ \> the set of all infinite sequences over the alphabet $\mathcal{X}$ \\[0.5ex]
  $\Delta \mathcal{X}$ \> a probability distribution over $\mathcal{X}$ \\[0.5ex]
  $O$ \> a reflective oracle \\[0.5ex]
  $\tilde{O}$ \> a partial oracle \\[0.5ex]
  $\flip$ \> the function that on input $p \in [0,1]$ returns 1 with probability $p$, else 0. \\[0.5ex]
  $\cQ_i$ \> a query to an oracle in an enumeration $i \in \mathcal{N}$ \\[0.5ex]
  $\cT$ \> the set of oracle probabilistic Turing machines, extended to include tape contents \\[0.5ex]
  $T$ \> an oracle probabilistic Turing machine \\[0.5ex]
  $\lambda_T$ \> the semimeasure corresponding to the probabilistic Turing machine $T$ 
  \\[0.5ex]
  $\lambda^O_T$ \> the semimeasure corresponding to the probabilistic Turing machine $T$ \\ \> with access to the reflective oracle $O$
  \\[0.5ex]
  $\bar{\lambda}^O_T$ \> the oracle-completed semimeasure corresponding to the probabilistic Turing machine $T$ \\ \> with access to the reflective oracle $O$
  \\[0.5ex]
  \rO \> abbreviation for ``reflective oracle'' \\[0.5ex]
  $O$-sampled \> a semimeasures $\nu$ is $O$-sampled if there exists a pTM $T$ such that $\nu = \lambda^O_T$ \\[0.5ex]
  $O$-estimable \> estimable with access to the reflective oracle $O$ \\[0.5ex]
  $\nu$ \> a semimeasure sometimes representing an environment \\[0.5ex]
  $\mu$ \> the true environment \\[0.5ex]
  $\mu \ll \nu$ \> $\mu$ is absolutely continuous w.r.t. $\nu$ \\[0.5ex]
  $\mathcal{A}$ \> a finite alphabet often identified with the set of possible actions \\[0.5ex]
  $\mathcal{E}$ \> a finite alphabet containing the possible percepts; rewards should be \\ \> computable from percepts \\[0.5ex]
  $\alpha,\beta$ \> alphabet symbols usually in $\mathcal{A}$ \\[0.5ex]
  $a_t$ \> the action(s) in time step $t$ \\[0.5ex]
  $e_t$ \> the percept(s) in time step $t$ \\[0.5ex]
  $\h_{<t}$ \> the first $t-1$ interactions, $a_1e_1...a_{t-1}e_{t-1}$ \\[0.5ex]
  $\epstr$ \> the empty string \\[0.5ex]
  $\eps$ \> a small positive real number \\[0.5ex]
  $\gamma$ \> the discount function $\gamma : \mathbb{N} \to R_{\geq 0}$ \\[0.5ex]
  $\Gamma_t$ \> a discount normalization factor $\Gamma_t = \sum^\infty_{k=t} \gamma_k$ \\[0.5ex]
  $\nu,\mu$ \> environments/semimeasures \\[0.5ex]
  $\sigma$ \> a multi-player game \\[0.5ex]
  $\sigma^{\pi_{1:n}}$ \> history distribution induced by $\pi_1, ..., \pi_n$ interacting in the multi-player game $\sigma$ \\[0.5ex]
  $\sigma_i$ \> the subjective environment of player $i$ in multi-player game $\sigma$ \\[0.5ex]
  $\pi$ \> a policy (strategy), $\pi: (\mathcal{A} \times \mathcal{E})^*\to\mathcal{A}$ \\[0.5ex]
  $V^\pi_\nu$ \> the $\nu$-expected value of policy $\pi$ \\[0.5ex]
  $V^*_\nu$ \> the optimal value in environment $\nu$ \\[0.5ex]
  $\cG$ \> a countable class of multi-player games \\[0.5ex]
  $\cM$ \> a countable class of environments \\[0.5ex]
  $\mathcal{M}\rOi$ \> the class of environments computed by probabilistic Turing machines \\ \> with access to reflective oracle $O$ \\[0.5ex]
  $\mathcal{P}$ \> a countable class of policies (strategies) \\[0.5ex]
  $\mathcal{P}\rOi$ \> the class of policies (strategies) computed by probabilistic Turing machines \\ \> with access to reflective oracle $O$ \\[0.5ex]
  $w$ \> a real valued weight distribution, usually over $\mathcal{M}\rOi$ or $\mathcal{P}\rOi$ \\[0.5ex]
  $\xi$ \> a mixture environment, usually over $\mathcal{M}\rOi$ \\[0.5ex]
  $\zeta$ \> a mixture policy over $\mathcal{P}\rOi$ \\[0.5ex]
\end{tabbing}

\section{Non-binary Alphabet Reflective Oracles}
\label{sec:non_binary_refl_oracle_existence}

We provide a detailed existence proof for non-binary reflective oracles. First, we extend the definition of reflective oracles to the non-binary case, which requires slightly stronger conditions than the binary case to ensure that the oracle completed conditionals lie on a probability simplex. Then we construct a point-to-set mapping that expresses an oracle's self-consistency and show that a fixed point would imply the existence of a reflective oracle. Finally, we demonstrate the existence of a fixed point using the infinite dimensional version of the Kakutani fixed point theorem \cite{fan1952fixed}. 

\paragraph{Definition.}\label{par:defs}
The most important feature of a reflective oracle is that it can be used to complete a semi-measure to a measure (which is then reflective-oracle sampled and even estimable). Extending the definition to non-binary alphabets, we must take care to preserve this property. It is easiest to do this by extending the strict definition of Fallenstein et.\ al.\ \cite{fallenstein_reflective_2015}, which we refer to as a ``step reflective oracle.'' Slightly generalizing their notation to explicitly include the input $x \in \{ 0, 1 \}^*$, they require that for every machine $T$ and input $x$, there is some cutoff $P[T(x) = 1] \leq q \leq 1 - P[T(x) \neq 1]$ such that
\[
p < q \Rightarrow O(T,x,p) = 1
\]
\[
p > q \Rightarrow O(T,x,p) = 0
\]
We will generalize to finite output alphabet $\cA$ and input alphabet $\Sigma$. For our purposes $\cA \subset \Sigma$; for instance $\cA$ is a set of actions and $\Sigma$ includes actions and observations.

\begin{definition}[reflective oracle]\label{def:reflective_oracle}
The oracle $O$ valued in [0,1] is reflective iff for each probabilistic TM $T$ and string $x \in \Sigma^*$, $\exists \{ q_\alpha \}_{\alpha \in \cA}$ satisfying the following properties:
\[
\sum_{\alpha \in \cA} q_\alpha ~=~ 1
\]
And for all $\alpha \in \cA$,
\[
P[T^O(x) = \alpha] ~\leq~ q_\alpha ~\leq~ 1 - P[T^O(x) \neq \alpha]
\]
\[
p < q_\alpha ~\Rightarrow~ O_\alpha(T,x,p) = 1
\]
\[
p > q_\alpha ~\Rightarrow~ O_\alpha(T,x,p)= 0 
\]
\end{definition}

Notice that oracles are now indexed by $\alpha$, so can be thought of as a family satisfying certain relations or as accepting a new argument of type $\cA$. An oracle responds to queries with 0 or 1, and the value of $O_\alpha$ on a given query is interpreted as the probability that its answer is 1. Given this definition, it is not difficult to see how to conduct a binary search for each $q_\alpha$ and recover a measure.

\paragraph{Proof.}
Now we must prove that non-binary reflective oracles actually exist. The argument follows the original existence proof. The key is to restrict certain functions to lie on the simplex (which is closed, compact, and convex) which makes it possible to ensure that the point-set map only accepts and produces functions satisfying the first condition. 

The argument relies on the infinite dimensional version of the Kakutani fixed point theorem (sometimes more appropriately but less helpfully referred to as the Kakutani-Ky Fan theorem), which we will reproduce from Ky Fan's paper \cite{fan1952fixed}:

\begin{theorem}[infinite dimensional Kakutani fixed point theorem]\label{kakutani_fixed_point}
    Let L be a locally convex topological linear space and $K$ a compact convex set in L. Let $R(K)$ be the family of all closed convex (non-empty) subsets of K. Then for any upper semicontinous point-to-set transformation $f$ from $K$ into $R(K)$, there exists a point $x_0 \in K$ s.t. $x_0 \in f(x_0)$. 
\end{theorem}

The more modern term for locally convex topogological linear space is locally convex topological vector space (LCTVS).

Let $\cT$ be the space of pTM's generalized to include the tape configuration. We will consider points $$(\query, \eval)$$ where $\eval \in (\Delta \cA)^\cT$, which we interpret as (for a fixed point) returning the completed chance of a given pTM outputting $\alpha$, and $\query \in [0,1]^{\cT \times (\mathbb{Q} \cap [0,1]) \times \cA}$ which is a function on queries returning the oracle's chance of producing a 1. We will later introduce constraints on $\query$ in terms of $\eval$ (through the point-to-set map $f$) so that a fixed point is a reflective oracle, which is why the definition of the space for $\query$ does not need to enforce that the oracle's answers describe a probability distribution. The space of $(\query, \eval)$ pairs is 
\[
K = [0,1]^{\cT \times (\mathbb{Q} \cap [0,1]) \times \cA} \times (\Delta \cA)^\cT 
\]
This is a subset (under a trivial ``currying'' homeomorphism on the space of $\eval$) of
\[
S = \mathbb{R}^{\cT \times (\mathbb{Q} \cap [0,1]) \times \cA} \times \mathbb{R}^{(\cT \times \cA)}
\]
which is an LCTVS under the product topology. This follows from the fact that $\mathbb{R}$ is a (very simple) LCTVS, and all LCTVS properties are closed under the product operation. 

\paragraph{Topological properties.}
We will later need that S is metrizable to use a notion of sequential convergence instead of ``upper semicontinuity". Because $\mathbb{R}$ is Hausdorff and products of Hausdorff spaces are also Hausdorff, S is a LCTVS in the stricter sense of Rudin's Functional Analysis \cite{rudin1991functional} (which requires $T_1$, a property weaker than Hausdorff). Theorem 9 tells us that TVS is metrizable if it has a countable local base. Our countable local base for $S$ is a sequence of balls with radius $1/n$ at the first $n$ points in enumerations of both the exponents ($\cT \times \cA$ and $\cT \times (\mathbb{Q} \cap [0,1]) \times \cA$), with $\mathbb{R}$ at the (infinitely many) remaining points. So $S$ is metrizable. Now we need to know that $K$ is compact and convex. Tychonoff's theorem implies it is compact. Convexity is immediate from convexity of simplices. 

\paragraph{The point-to-set map $f$.}
The original proof demonstrates that it is possible to construct an oracle reflective on some subset of queries and equal to a different fixed oracle elsewhere. This is more than we need and clutters the proof, but is easy to understand once the proof is digested, so we focus on the case that we want a reflective oracle on all queries. We are ready to define our point-to-set map $f : K \to 2^K$. We will then demonstrate that the range of $f$ is $R(K)$. Let $(\query, \eval) \in f((\query, \eval))$ iff the following conditions hold:
\[
p < \eval_\alpha(T) \Rightarrow \query'_\alpha(T,p) = 1
\]
\[
p > \eval_\alpha(T) \Rightarrow \query'_\alpha(T,p) = 0
\]
This leaves $\query'_\alpha(T, \eval_\alpha(T))$ unconstrained.

The recursive rules on $\eval'$ are slightly more complicated. The ``base case'' is that if $T$ returns $T() = \beta$ on its next computation step,
\[
\eval'_\alpha(T) = \delta_\alpha(T()) = [\![\alpha = \beta]\!]
\]
$T$ can halt returning nothing or something other than a single symbol. In that case, $\eval'_\alpha(T)$ is arbitrary.  

The other ``inductive'' rules are copied from \cite{fallenstein_reflective_2015}. If $T$ performs a deterministic computation step producing a new machine/configuration $N$, $\eval'_\alpha(T) = \eval(N)$. If $T$ performs a coin flip yielding a state $N$ with rational probability $p$ and $N'$ with rational probability $1-p$, $\eval'_\alpha(T) = p \eval_\alpha(N) + (1-p) \eval_\alpha(N')$. We observe that typically a transition of a pTM is defined to use exactly one random bit so we should have $p = \frac{1}{2}$. Also, if $T$ has a chance of halting with some output immediately after reading the random bit (in the same computation step), $\eval_\alpha(N)$ or $\eval_\alpha(N')$ should be replaced according to the first rule; otherwise they are not really well-defined. Calls to the oracle should be treated the same as coin flips with probability determined by $\query$: if the oracle is invoked on $(T,p)$ yielding $N$ on 1 and $N'$ on 0, and $\query(T,p) = q$ then $\eval'_\alpha(T) = q \eval_\alpha(N) + (1-q) \eval_\alpha(N')$.

Assuming for a moment that a fixed point exists, we will show that this gives us a reflective oracle, defined by $P[O_\alpha(T, p)=1] = \query_\alpha(T, p)$. We normally want our reflective oracles to accept a machine description and input string, but these can be combined into an extended machine/configuration as in the definition of $O$.

By induction on the number of computation steps,
\[
P[T^O() = \alpha] \leq \eval_\alpha(T) \leq 1- P[T^O() \neq \alpha]
\]
Together with the conditions of $\query_\alpha(T,p)$, this shows that $\eval_\alpha(T)$ satisfies the last three conditions on $q_\alpha$ in \cref{par:defs}. The first condition is automatically satisfied because $\eval(T)$ is restricted to the simplex; this pushes the burden of proof to the non-emptiness of $f((\query, \eval))$. Therefore, $O$ is a reflective oracle.

\paragraph{Existence of a fixed point.}
Now we only need to prove the existence of a fixed point by establishing the conditions of the infinite dimensional version of the Kakutani fixed-point theorem. First, we will show that $f((\query, \eval))$ is closed, convex, and non-empty. In a metrizable space, closed sets can be characterized in the ordinary way by sequential convergence. Note that at every point, $\query'$ and $\eval'$ are either restricted to some fixed value depending on $(\query, \eval)$ or are unrestricted (except for $\eval'$ to lie on the simplex, which is closed). Taking limits (of any such sequence) shows that $f((\query, \eval))$ is closed. Convexity is also easy to verify pointwise (and from convexity of the simplex). Finally, we must show non-emptiness. It is obvious that the conditions on $\query'$ can always be satisfied. It remains only to show that the conditions on $\eval'$ can be satisfied by points lying on the simplex. But this is only another application of convexity of the simplex, noting that $q$ and $p$ lie in [0,1], and $\eval(T) \in \Delta \cA$, and verifying that the ``base case'' condition produces a point on the simplex (in fact an extreme point of the simplex).  

The last property to verify is ``upper semicontinuity.'' But as Fan points out, this is equivalent to the definition in terms of convergent sequences, which is usually called the ``closed graph property.'' We must show that if $(\query'_n, \eval'_n) \to (\query', \eval')$, $(\query_n, \eval_n) \to (\query, \eval)$, and $(\query'_n, \eval'_n) \in f((\query_n, \eval_n))$, then $(\query', \eval') \in f((\query, \eval))$. 

The argument is exactly the same as in the binary case, but applies ``pointwise'' to each $\alpha$. Taking limits on both sides of the rules for $\eval'$ immediately gives the desired result for this part. For $\query'_\alpha(T,p)$, the argument depends on the relationship between $\eval_\alpha(T)$ and $p$. If they are equal, the condition is automatically satisfied. The two remaining cases are symmetric, so we consider (w.l.o.g.) the case that $\eval_\alpha(T) > p$. Since $\eval_n\to\eval$ (in the topology of pointwise convergence), $(\eval_{n})_\alpha(T)\to\eval_\alpha(T)$, and for sufficiently large $n$, $(\eval_n)_\alpha(T) > p$. This means that $(\query'_n)_\alpha(T, p) \to 1$, so $\query'_\alpha(T,p) = 1$. This proves the closed graph property, which implies a fixed point of $f$ by the infinite dimensional Kakutani fixed point theorem. Therefore, there is a non-binary alphabet reflective oracle.

\paragraph{Types for symbols.}
Sometimes the alphabet $\cA = \bigsqcup_{1\leq i \leq n} \cA_i$ and each machine has an intended \emph{type} in $\{1,..,n\}$. Then we want to interpret our machines as semimeasures over the corresponding $\cA^*_i$ and complete the conditionals to $\Delta \cA_i$. The primary examples are when each player has a different action set (so $n$ is the number of players) and when a player interacts with an environment (so $n=2$ for the action and percept spaces). The natural idea is to use a different reflective oracle with each output alphabet $\cA_i$, but unfortunately the oracles need to be answer questions about each other's behavior so this does not literally work. Instead we can change the requirement $\sum_{\alpha \in \cA} q_\alpha = 1$ to the $n$ requirements $\sum_{\alpha \in \cA_i} q_\alpha = 1$, and force the oracle to satisfy this by changing the space of $\eval$ from $(\Delta \cA)^\cT$ to $(\prod_{i=1}^n \Delta \cA_i)^\cT$. It is easy to see that this set is still closed, convex, and compact. The rest of the existence proof goes through essentially unchanged. The resulting (typed) non-binary oracle automatically knows which output type we expect for a machine based on the type of its symbol argument $\alpha \in \cA_i$ and completes the associated semimeasure appropriately, redistributing the probability of outputs outside of $\cA_i$ in the same way as non-halting probability mass.    

\paragraph{Reflectivity on subsets.}
It is easy to modify the proof above so that we construct an oracle $O$ that satisfies reflectivity on some subset of queries $R$ but behaves identically to any arbitrary oracle $O'$ outside of $R$. Let $A$ be a \emph{closed} set of pTM's that only make oracle calls about other pTM's in $A$. Let $R = \{(T,x,p)|T \notin A\}$. Then if $O'$ is reflective on queries about pTM's in $A$ (that is, $R^C$), we can construct an oracle $O$ reflective on $R$ that agrees with $O'$ on $R^C$. But because replacing $O'$ by $O$ does not change the results of any oracle calls for machines in $A$, $O$ is also reflective on $R^C$, which means $O$ is a reflective oracle (on all queries). This is analogous to the extending a linearly indepedent set of vectors to a basis. In some cases it is easy to find a variety of explicit oracles $O'$ reflective on $A$; for instance, if $A = \{ T \}$ where $T$ never halts, $O'$ can complete $\lambda^O_T$ to a measure in any arbitrary way, and there will exist a reflective oracle $O$ agreeing with $O'$ about $T$. This means that all measures are reflective-oracle computable with appropriate choice of $O$ (though by countability of $\cT$ there is no \emph{particular} $O$ that makes every measure \rOc).   

\section{Limit-Computable Step Reflective Oracles} \label{sec:limit_computability}

We will extend Leike's proof that there is a limit-computable reflective oracle \cite{leike_formal_2016} to show that there is a limit-computable non-binary reflective oracle. This requires a slightly different construction which yields a limit-computable step reflective oracle when restricted to the binary case, simplifying the oracle completion process by removing the need for an expectation. 

Following \cite{leike_formal_2016}, we will construct an infinite sequence of partial oracles converging to a reflective oracle in the limit. The set of queries to a reflective oracle is countable and computably enumerable, so we will fix a computable enumeration:
\[
\cT \times \Sigma^* \times \mathbb{Q} =: \{ \cQ_1, \cQ_2, ... \}
\]

where $\cT$ is the set of (generalized) pTM's as above and $\Sigma$ is the input alphabet. A reflective oracle is also indexable by symbols from the output alphabet. 

\begin{definition}[k-partial oracle]
    A k-partial oracle $\tilde{O}_\alpha$ is a function from the first k queries to the multiples of $2^{-k}$ in [0,1]:
    \[
    \tilde{O}_\alpha : \{\cQ_1, \cQ_2, ..., \cQ_k \}\to\{ n2^{-k} | 0 \leq n \leq 2^k \}
    \]
\end{definition}

\begin{definition}[approximate an oracle]
    A k-partial oracle $\tilde{O}$ approximates an oracle $O$ iff $\forall \alpha |O_\alpha(\cQ_i) - \tilde{O}_\alpha(\cQ_i)| \leq 2^{-k-1}$ for all $i \leq k$.
\end{definition}

Let $\tilde{O}$ be a $k$-partial oracle for $k \in \mathbb{N}$ and let $T \in \cT$ be an oracle machine. We define $T^{\tilde{O}}$ to be the following machine:

\begin{enumerate}
    \item Run $T$ for at most $k$ steps.
    \item If $T$ makes an oracle $\alpha$ call on $\cQ_i$ for $i \leq k$,
    \begin{enumerate}
        \item Return 1 with probability $\tilde{O}_\alpha(\cQ_i) - 2^{-k}$
        \item Return 0 with probability $1 - \tilde{O}_\alpha(\cQ_i) - 2^{-k}$
        \item halt otherwise
    \end{enumerate} 
    \item If $T$ calls the oracle on $\cQ_j$ for $j > k$, halt.
\end{enumerate}

Since $\tilde{O}$ is not a fully defined oracle, this is different than the usual meaning of $T^O$. In particular it implicitly depends on the value of $k$.

\begin{lemma}[bound on $\lambda_T^{\tilde{O}}$] \label{lambda_bounds}
    If a $k$-partial oracle $\tilde{O}$ approximates a reflective oracle $O$, then $\lambda_T^O(\alpha|x) \geq \lambda_T^{\tilde{O}}(\alpha|x)$ for all $\alpha \in \cA$, $x \in \Sigma^*$, and $T \in \cT$. 
\end{lemma}

\begin{proof} This follows from the definition of $T^{\tilde{O}}$: when running $T$ with $\tilde{O}$ instead of $O$, every sequence of oracle responses is less likely because $\tilde{O}_\alpha - 2^{-k} < O_\alpha(\cQ_i)$ and $1 - \tilde{O}_\alpha - 2^{-k} < 1 - O(\cQ_i)$. The other differences can also only lose probability mass. If $T$ makes calls whose index is $> k$ or runs for more than $k$ steps the machine halts and no output is generated. 
\end{proof}

\begin{definition}[$k$-partially reflective]
    A $k$-partial oracle $\tilde{O}$ is $k$-partially (step) reflective iff for the first $k$ queries $(T,x,p)$ and for all $\alpha$
    \begin{enumerate}
        \item $p < \lambda_T^{\tilde{O}}(\alpha|x)$ implies $\tilde{O}_\alpha(T,x,p) = 1$, and
        \item $p > 1 - \sum_{\beta \neq \alpha} \lambda_T^{\tilde{O}}(\beta|x)$ implies $\tilde{O}_\alpha(T, x, p) = 0$.
    \end{enumerate}
    Also, we require that for all $\alpha, T, x$, $\tilde{O}_\alpha(T,x,\cdot)$ is non-increasing and there is at most one $p$ such that $(T,x,p)$ is in the first $k$ queries and $\tilde{O}_\alpha(T,x,p) \notin \{ 0, 1 \}$. 

    Finally, for each $T,x$ appearing in one of the first $k$ queries,   
    \begin{enumerate}
        \item $\sum_{\alpha \in \cA} \min \{p | (T,x,p) \in \{Q_i\}_{i \leq k} \text{\ and\ } \tilde{O}_\alpha(T,x,p) = 0\} \geq 1$
        \item $\sum_{\alpha \in \cA} \max \{p | (T,x,p) \in \{Q_i\}_{i \leq k} \text{\ and\ } \tilde{O}_\alpha(T,x,p) = 1 \} \leq 1$
    \end{enumerate}
    
\end{definition}
The minima default to 1 and the maxima default to 0. The first pair of conditions enforce the reflective oracle property. The remaining conditions enforce the step reflective oracle property which is always required for non-binary reflective oracles.

We can check whether a $k$-partial oracle is $k$-partially reflective in finite time. The first pair of conditions can be checked by running the machines from the first $k$ queries for $k$ steps each (on every combination of the $\leq 2^k$ random bits used) and calculating $\lambda_T^{\tilde{O}}(\alpha | x)$ exactly. The rest are clearly possible to verify with one pass over the first $k$ queries for each $\alpha$. 

\begin{lemma}[partial approximations are partially reflective]\label{lemma:approx_refl_implies_partial_refl}
    If $O$ is a reflective oracle and $\tilde{O}$ is a $k$-partial oracle that approximates $O$, then $\tilde{O}$ is $k$-partially reflective.
\end{lemma}

\begin{proof} Note that since $\tilde{O}$ assigns values in a $2^{-k}$ grid and approximates $O$ up to $2^{-k-1}$, for the first $k$ queries $O_\alpha(T,x,p) = 0\to\tilde{O}_\alpha(T,x,p) = 0$ and $O_\alpha(T,x,p) = 1\to\tilde{O}_\alpha(T,x,p) = 1$. Assuming $\lambda^{\tilde{O}}_T(\alpha | x) > p$ we get from \cref{lambda_bounds} that $\lambda_T^O(\alpha|x) \geq \lambda_T^{\tilde{O}}(\alpha|x) \geq p$ so $1 = O_\alpha(T,x,p) = \tilde{O}_\alpha(T,x,p)$. The second condition is proved symmetrically. For the remaining conditions, recall that for each $T,x$, $\exists q_\alpha$ with $O(T,x,p) = 1$ for $p < q_\alpha$, $O(T,x,p) = 0$ for $p > q_\alpha$, and $\sum_\alpha q_\alpha = 1$. This means that $\tilde{O}_\alpha(T,x,p) = 1$ for $p < q_\alpha$ and $\tilde{O}_\alpha(T,x,p) = 0$ for $p > q_\alpha$, and since these are the maximum and minimum values of $\tilde{O}$ it is certainly non-increasing regardless of its value at $q_\alpha$ (which is in general not defined). Additionally, if $\tilde{O}_\alpha(T,x,p) = 0$ then certainly $p \geq q_\alpha$, which implies that the sum of minima is $\geq \sum_\alpha q_\alpha = 1$ (this is also automatically satisfied if any $\tilde{O}_\alpha(T,x,\cdot)$ does not take the value 0). The argument for the bound on the maxima is similar. 
\end{proof}

\begin{definition} [extending partial oracles] \label{def:extending_partial_O}
    A $k+1$ partial oracle $\tilde{O}'$ extends a $k$-partial oracle $\tilde{O}$ iff $|\tilde{O}_\alpha(\cQ_i) - \tilde{O}'_\alpha(\cQ_i)| \leq 2^{-k-1}$ for all $i \leq k$. 
\end{definition}

\begin{lemma}[infinite sequence of extensions]\label{lemma:infinite_seq}
    There is an infinite sequence of partial oracles $(\tilde{O}^k)_{k \in \mathbb{N}}$ such that for each $k$, $\tilde{O}^k$ is a $k$-partially reflective $k$-partial oracle and $\tilde{O}^{k+1}$ extends $\tilde{O}^k$.
\end{lemma}

\begin{proof} As shown in \cref{sec:non_binary_refl_oracle_existence}, there is a (step) reflective oracle $O$ for any finite alphabet.
For every $k$, there is a canonical $k$-partial oracle $\tilde{O}_k$ that approximates $O$: restrict $O$ to the first $k$ queries and for any such query $\cQ$ for each $\alpha \in \cA$ pick the value in the $2^{-k}$ grid which is closest to $O_\alpha(\cQ)$. By construction, each $\tilde{O}^{k+1}$ extends $\tilde{O}^k$ and by \cref{lemma:approx_refl_implies_partial_refl}, each $\tilde{O}^k$ is $k$-partially reflective.
\end{proof}

\begin{lemma}[$\lambda_T^{\tilde{O}^{k}}$ increases]\label{lemma:lambdas_increase}
    If the $k+1$ partial oracle $\tilde{O}^{k+1}$ extends the $k$-partial oracle $\tilde{O}_k$, then $\forall \alpha$ $\lambda_T^{\tilde{O}^{k+1}}(\alpha | x) \geq \lambda_T^{\tilde{O}^{k}}(\alpha | x)$ for each $T \in \cT$ and $x \in \cA^*$ 
\end{lemma}

\begin{proof} $T^{\tilde{O}^{k+1}}$ runs for one more step than $T^{\tilde{O}^k}$ and can answer one more query. Because $\tilde{O}^{k+1}$ extends $\tilde{O}^k$, $|\tilde{O}^{k+1}_\alpha(\cQ_i)-\tilde{O}^k_\alpha(\cQ_i)| \leq 2^{-k-1}$, which means $\tilde{O}^{k+1}_\alpha(\cQ_i) - 2^{-k-1} \geq \tilde{O}^k_\alpha(\cQ_i) - 2^{-k}$, so halting on oracle calls is less likely and the chances of returning 0 or 1 are both higher.  
\end{proof}

\paragraph{Search algorithm.}
Now we are prepared to state the algorithm that constructs a reflective oracle in the limit. The algorithm recursively traverses a directed acyclic graph (DAG) of partial oracles. The DAG's nodes are the partial oracles; level $k$ of the DAG contains all $k$-partial oracles. There is an edge in the DAG from the $k$-partial oracle $\tilde{O}^k$ to the $i$-partial oracle $\tilde{O}^i$ if and only if $i = k+1$ and $\tilde{O}^i$ extends $\tilde{O}^k$.

For every $k$, there are only finitely many $k$-partial oracles, since they are functions from finite sets to finite sets. In particular, there are exactly two 1-partial oracles (so the DAG has two \emph{sources}, nodes without parents where the search can begin). Pick one of them to start with, and proceed recursively as follows. Given a $k$-partial oracle $\tilde{O}^k$, there are finitely many $(k+1)$-partial oracles that extend $\tilde{O}^k$ (finite out-degree). Pick one that is $(k+1)$-partially reflective (which can be checked in finite time). If there is no $(k+1)$-partially reflective extension, backtrack.

By \cref{lemma:infinite_seq} our DAG is infinitely deep and thus the search does not terminate. Moreover, it can backtrack to each level only a finite number of times because there are only a finite number of paths from a source to each level and at each level there are only a finite number of possible extensions (in fact, though it is possible for two different paths from sources in the DAG to reach the same node, there is no need to return to any node once it has been visited and all of its children have been explored). Therefore, the algorithm will produce an infinite sequence of partial oracles, each extending the previous. Because of finite backtracking, the output eventually stabilizes on a sequence of partial oracles $\tilde{O}^1, \tilde{O}^2, ...$. By the following lemma, this sequence converges to a reflective oracle, proving \cref{thm:non_binary_limit_comp_refl_oracle}. 

\begin{lemma}[limit is reflective]\label{lemma:limit_is_refl_oracle}
    Let $\tilde{O}^1, \tilde{O}^2, ...$ be a sequence where $\tilde{O}^k$ is a $k$-partially reflective oracle and $\tilde{O}^{k+1}$ extends $\tilde{O}^k$ for all $k \in \mathbb{N}$. Let $O := \lim_{k\to\infty} \tilde{O}^k$ be the pointwise limit. Then
    \begin{enumerate}
        \item $\lambda_T^{\tilde{O}^k}(\alpha|x)\to\lambda_T^O(\alpha|x)$ as $k\to\infty$ for all $\alpha \in \cA$ and $x \in \Sigma^*$. 
        \item $O$ is a reflective oracle.
    \end{enumerate}
\end{lemma}

\begin{proof} First note that the pointwise limit must exist because $|\tilde{O}^k_\alpha(\cQ_i) - \tilde{O}^{k+1}_\alpha(\cQ_i)| \leq 2^{-k-1}$ by \cref{def:extending_partial_O}. 

\begin{enumerate}
    \item Since $\tilde{O}^{k+1}$ extends $\tilde{O}^k$, each $\tilde{O}^k$ approximates $O$. Let $\alpha \in \cA, T \in \cT,$ and $x \in \Sigma^*$ and consider the sequence $a_k := \lambda_T^{\tilde{O}^k}(\alpha|x)$. By \cref{lemma:lambdas_increase}, $a_k \leq a_{k+1}$ so the sequence is monotonically increasing. It is also bounded above by $\lambda^O_T(\alpha|x)$ according to \cref{lambda_bounds}, so it converges. It only remains to show that the sequence does not converge to something less than $\lambda^O_T(\alpha|x)$. The probability that $T^O$ halts, which is bounded above by 1, is the sum of its probabilities of halting at each step. Therefore, the probability $T^O$ halts after running for more than $k$ steps is a tail sum that approaches 0 as $k\to\infty$. The distributions on the results of calls to the partial oracles converge to the distribution on the results of calls to $O$ by definition of $O$ and because $2^{-k} \to 0$ in the definition of $T^{\tilde{O}^k}$. The definition of $T^O$ therefore implies that $a_k\to\lambda_T^O(\alpha|x)$ as desired. \label{limit_claim}
    
    \item By definition, $O$ is an oracle. 
    It only remains to show that $O$ satisfies the step reflective conditions given in \cref{par:defs}. Consider a fixed $T \in M$ and $x \in \Sigma^*$. Let $P^k = \{ p \in \mathbb{Q} | (T,x,p) \in \{\cQ_1, \cQ_2, ... \cQ_k \} \}$. Let $L^k_\alpha$ be the subset of $P^k$ for which $\tilde{O}^k_\alpha(T,x,p) = 1$ and $H^k_\alpha$ be the subset of $P^k$ for which $\tilde{O}^k_\alpha(T,x,p) = 0$. Because $\tilde{O}^k$ is $k$-partially reflective, there may exist at most one point $p_\alpha$ which is in $P^k$ but is not in $L^k_\alpha$ or $H^k_\alpha$. Because $\tilde{O}^k$ takes values on a $2^{-k}$ grid and $\tilde{O}^{k+1}$ extends $\tilde{O}^k$, $p_\alpha$ is not in $L^t_\alpha$ or $H^t_\alpha$ for any $k' \geq k$; it is not possible to reach 0 or 1 from $\tilde{O}^k_\alpha(T,x,p_\alpha)$ by extension since
    \begin{equation}
        \begin{split}
            |\tilde{O}^{k'}_\alpha(T,x,p_\alpha) - \tilde{O}^k_\alpha(T,x,p_\alpha)| & \leq \sum_{i=k}^{k'-1} |\tilde{O}^{i+1}_\alpha(T,x,p_\alpha) - \tilde{O}^i_\alpha(T,x,p_\alpha)| \\
            & \leq \sum_{i=k}^{k'-1} 2^{-i-1} < \sum_{i=k}^\infty 2^{-i-1} = 2^{-k}
        \end{split}
    \end{equation}
    This means that if any such $p_\alpha$ exists it does not depend on $k$, and we will consider only $t$ sufficiently large so that $p_\alpha$ is in $P^k$ but not in $L^k_\alpha$ or $U^k_\alpha$ for all $k \geq t$. We will assume without loss of generality that $p_\alpha$ exists; otherwise the proof is routinely simplified. By uniqueness, no point other than $p_\alpha$ can lie in $P^k - (L^k_\alpha \cup H^k_\alpha)$ for any such $k$. Using the bounds on extensions again, $L^k_\alpha \subseteq L^{k+1}_\alpha$ and $H^k_\alpha \subseteq H^{k+1}_\alpha$. Noting that $k$-partial step reflectivity of each $\tilde{O}^k$ requires it is non-increasing, we must have $L^k_\alpha < H^k_\alpha$ element-wise, so $\sup \bigcup_{k \geq t} L^k_\alpha \leq \inf \bigcup_{k \geq t} H^k_\alpha$. Additionally, $\bigcup_{k \geq t} P^k = \mathbb{Q} \cap [0,1]$ since all queries are enumerated. This implies that $(\bigcup_{k \geq t}L^k_\alpha) \cup (\bigcup_{k \geq t} H^k_\alpha) = \mathbb{Q} \cap [0,1] - \{ p_\alpha \}$. Therefore,
    \[
    \lim_{k\to\infty} \max L^k_\alpha = \sup \bigcup_{k \geq t} L^k_\alpha = \inf \bigcup_{k \geq t} H^k_\alpha = \lim_{k\to\infty} \min H^k_\alpha
    \]
    The non-increasing property also requires that these limits are $p_\alpha$ (in the case that $p_\alpha$ does not exist we use this as its definition). Because $H^k_\alpha$ are increasing sets and eventually include all $p < p_\alpha$, for such $p$, $1 = \lim_{k\to\infty} \tilde{O}^k(T,x,p) = O(T,x,p)$. Similarly, $O(T,x,p) = 0$ for $p > p_\alpha$. It is also clear that $\lambda_T^O(\alpha|x) \leq p_\alpha \leq 1 - \sum_{\beta \neq \alpha} \lambda_T^O(\beta|x)$ because the bounds are the limits of $\lambda_T^{\tilde{O}^k}(\alpha|x)$ and $1 - \sum_{\beta \neq \alpha} \lambda_T^{\tilde{O}^k}(\beta|x)$ and $\tilde{O}^k$ is $k$-partially reflective. All that remains to show is that $\sum_\alpha p_\alpha = 1$. But the last pair of requirements for $k$-partial reflectivity are that $\sum_\alpha \min H^k_\alpha \geq 1$ and $\sum_\alpha \max L^k_\alpha \leq 1$. Taking the limits of both sides, $1 \leq \sum_\alpha p_\alpha \leq 1$, so $\sum_\alpha p_\alpha = 1$. Therefore, $O$ is a reflective oracle.
\end{enumerate}
\end{proof}
%
\begin{theorem}[limit-computable non-binary reflective oracle]\label{thm:non_binary_limit_comp_refl_oracle}
    There is a limit-computable reflective oracle over any finite output alphabet. 
\end{theorem}
\begin{proof} The search algorithm produces a reflective oracle in the limit by \cref{lemma:limit_is_refl_oracle}.
\end{proof}

\section{General Reflective Oracle Computability of Completed Semimeasures} \label{sec:general_refl_oracle_comp}

We have restricted our focus to step reflective oracles, but \cite{leike_formal_2016} uses a more general class (defined only for binary alphabets) that can randomize at any point between $\lambda_T^O(1|x)$ and $1 - \lambda_T^O(0|x)$. The oracle no longer needs to be indexed by symbol because the completed probability of 0 can be determined from the completed probability of 1. With this definition, there is not necessarily a unique crossover point $q$ where $O(T,x,\cdot)$ switches from 1 to 0. However, $\bar\lambda_T^O(\cdot | x)$ can still be defined by running a binary search. The only complication is that the limit is no longer deterministic. 

\begin{algorithm} 
	\caption{pTM $C_T$}
    \label{alg:C_T}
	\begin{algorithmic}[1]
        \Input $x \in \{ 0, 1 \}^*$
        \Require Random bit sequence $\omega$
        \Output $y \sim \bar{\lambda}^O_T(\cdot|x)$
        \State $l, h = 0, 1$
        \For{$i=1,2,...$}
            \State $m = \frac{l + h}{2}$
            \If{$\flip(O(T, x, m))$} $l \gets m$
            \Else \ $h \gets m$
            \EndIf 
            \If{$\omega_{1:i} + 2^{-i} < l$} Return 1
            \ElsIf{$\omega_{1:i} > h$} Return 0
            \EndIf
        \EndFor
    \end{algorithmic}
\end{algorithm}

Let $\bar\lambda_T^O = \lambda_{C_T}^O$, defined in \ref{alg:C_T}. On input string $x \in \{ 0, 1 \}^*$, let $p^*$ be the limit point of the binary search using $O(T,x,\cdot)$. This is a random variable depending on the stochasticity of queries to $O$. Fixing the oracle's random choices, $p^*$ is the probability (over random bits $\omega$) that $C_T$ returns $1$. This implies that the overall probability that $C_T$ returns $1$ is $\bar\lambda_T^O = \lambda_{C_T}^O = \mathbb{E}[p^*]$ (with the expectation over the oracle's responses). Because $C_T$ halts with probability $1$, $\bar\lambda_T^O$ is a measure and is $O$-estimable (this time with a deterministic binary search). 

\section{Lower Semicomputability of pTM sampled Semimeasures}\label{sec:lscm_vs_pTM}

\begin{theorem}[l.s.c.\ of pTM semimeasures]\label{thm:lscsm_from_pTM}
    For any pTM $T$, $\lambda_T$ has l.s.c.\ conditionals.
\end{theorem}

\begin{proof} We can see this by constructing a binary tree with each edge corresponding to the next random bit received by $T$. For some sequences of random bits, $T$ halts after reading a prefix without requesting any further bits. If we mark the leaves of this tree with the output of $T$ when it halts (which is deterministic once the bits are fixed) then $\lambda_T(\alpha|x)$ is the sum of the probabilities for each random string along a path from the root to a leaf labeled with $\alpha$, which is $2^{-l}$ for a path of length l. Though there are uncountably many infinite sequences of random bits, so we cannot literally run $T$ on all random bit sequences in parallel, we can run a breadth-first search on the binary tree to increasing depths and sum the probabilities of each leaf marked $\alpha$ encountered, so $\lambda_T$ is l.s.c.\ 
\end{proof}

\begin{theorem}[l.s.c.\ conditionals are sampled by a pTM]\label{thm:pTM4lscsm}
    If $\mu$ has l.s.c.\ conditionals, we can find a pTM $T$ such that $\mu = \lambda_T$. 
\end{theorem}

\begin{proof} Let $\phi_\alpha(x, k)$ lower semicompute $\mu(\alpha|x)$ with $O$ access. 
The trick (similar to the proof of the coding theorem \cite{livitanyi}) is to partition the interval into a list P of subintervals, each labeled with a symbol from $\cA$. We will define some subroutines for manipulating P by adding a new subinterval on the right of a given length and checking whether a point is in a labeled subinterval. See \cref{alg:sample} for details.

\noindent\begin{minipage}[t]{0.48\textwidth}
\begin{algorithm}[H]
	\caption{push} 
	\begin{algorithmic}[1]
        \Input P, $\alpha \in \cA$, $\Delta \in \mathbb{Q}$ 
        \Effect A subinterval of label $\alpha$ and length $\Delta$ is added to P
        \State left $\gets$ P[-1].right
        \State P.append(label: $\alpha$, right: left + $\Delta$)
    \end{algorithmic}
\end{algorithm}\vspace{-3ex}
\begin{algorithm}[H]
	\caption{check} 
	\begin{algorithmic}[1]
        \Input P and $\omega$
        \Effect If $\omega$ is in a subinterval, return its label
        \State left $\gets$ 0
        \For{each label, right $\in$ P}
            \If{left $\leq \omega \leq$ right}
                \State return label
            \EndIf
            \State left $\gets$ right
        \EndFor
    \end{algorithmic}
\end{algorithm}
\end{minipage}
\hfill
\begin{minipage}[t]{0.48\textwidth}
\begin{algorithm}[H]
	\caption{sample}
    \label{alg:sample}
	\begin{algorithmic}[1]
        \Input $\phi_\alpha, x$ 
        \Require random stream $\omega$
        \Output $\alpha \sim \mu(\cdot|x)$
        \State Let P be an empty list
        \State $\psi_\alpha \gets 0$
        \For{$k \gets 1$ to $\infty$}
            \For{$\alpha \in \cA$}
                \State $\Delta \gets \phi_\alpha(x,k) - \psi_\alpha$
                \State $\psi_\alpha \gets \phi_\alpha(x,k)$
                \State push(P,$\alpha$,$\Delta$)
            \EndFor
            \State check(P,$\omega_{1:k})$
        \EndFor
    \end{algorithmic}
\end{algorithm}
\end{minipage}

Note that $w_{1:k}$ is only specified to precision $2^{-k}$. We extend the $\leq$ and $\geq$ comparisons against it to only succeed when they succeed in the worst case. It is fairly easy to see that in the limit the total area of the partitions for each $\alpha$ is equal to $\mu(\alpha|x)$. Because $\mu$ is only assumed to be a semimeasure, it is possible that some of the interval is unallocated and in this case \cref{alg:sample} may not halt. 
\end{proof}

\section{The Subjective Environment is Well-Defined}\label{app:subjective_environment}

For the subjective environment $\sigma_i$ to qualify as a true environment, it must not depend on the strategy $\pi_i$. With a bit of algebra we can show that it depends on $\sigma$ and $\pi_j$ for $j \neq i$ but not $\pi_i$:
\[
\sigma_i(e^i_T|\text{\ae}^i_{<T}a^i_T) = \frac{\sum_{\text{\ae}^j_{\leq T},\ j \neq i} \sigma^\pi(\text{\ae}_{\leq T})}{\sum_{\text{\ae}^j_{< T}a^j_T,\ j \neq i} \sigma^\pi(\text{\ae}_{< T}a_T)} 
\]
\[
= \frac{\sum_{\text{\ae}^j_{\leq T},\ j \neq i} \prod_{t=1}^T \sigma^\pi(e_t|\text{\ae}_{<t}a_t) \prod_{t=1}^T \sigma^\pi(a_t|\text{\ae}_{<t})}{\sum_{\text{\ae}^j_{< T}a^j_T,\ j \neq i} \prod_{t=1}^{T-1} \sigma^\pi(e_t|\text{\ae}_{<t}a_t) \prod_{t=1}^T \sigma^\pi(a_t|\text{\ae}_{<t})}
\]
\[
= \frac{\sum_{\text{\ae}^j_{\leq T},\ j \neq i} \prod_{t=1}^T \sigma(e_t|\text{\ae}_{<t}a_t) \prod_{t=1}^T \prod_{j=1}^n \pi_j(a^j_t|\text{\ae}^j_{<t})}{\sum_{\text{\ae}^j_{< T}a^j_T,\ j \neq i} \prod_{t=1}^{T-1} \sigma(e_t|\text{\ae}_{<t}a_t) \prod_{t=1}^T \prod_{j=1}^n \pi_j(a^j_t|\text{\ae}^j_{<t})}
\]
\[
= \frac{\prod_{t=1}^T \pi_i(a^i_t|\text{\ae}^i_{<t}) \sum_{\text{\ae}^j_{\leq T},\ j \neq i} \prod_{t=1}^T \sigma(e_t|\text{\ae}_{<t}a_t) \prod_{t=1}^T \prod_{j \neq i} \pi_j(a^j_t|\text{\ae}^j_{<t})}{\prod_{t=1}^T \pi_i(a^i_t|\text{\ae}^i_{<t}) \sum_{\text{\ae}^j_{< T}a^j_T,\ j \neq i} \prod_{t=1}^{T-1} \sigma(e_t|\text{\ae}_{<t}a_t) \prod_{t=1}^T \prod_{j \neq i} \pi_j(a^j_t|\text{\ae}^j_{<t})}
\]
\[
= \frac{ \sum_{\text{\ae}^j_{\leq T},\ j \neq i} \prod_{t=1}^T \sigma(e_t|\text{\ae}_{<t}a_t) \prod_{t=1}^T \prod_{j \neq i} \pi_j(a^j_t|\text{\ae}^j_{<t})}{\sum_{\text{\ae}^j_{< T}a^j_T,\ j \neq i} \prod_{t=1}^{T-1} \sigma(e_t|\text{\ae}_{<t}a_t) \prod_{t=1}^T \prod_{j \neq i} \pi_j(a^j_t|\text{\ae}^j_{<t})}
\]
As desired, all appearances of $\pi_i$ cancel.

\end{document}